\documentclass[11pt]{article}
\usepackage[a4paper,left=20mm,top=20mm,right=20mm,bottom=25mm]{geometry}

\usepackage{authblk}
\usepackage{natbib}

\usepackage{textcomp}
\usepackage{textgreek}
\usepackage[colorlinks=true,urlcolor=blue,citecolor=blue,linkcolor=blue]{hyperref}
\usepackage{graphicx}
\usepackage{color}
\usepackage{float}
\usepackage{bm}
\usepackage[noend]{algorithmic}
\usepackage{algorithm}
\usepackage{blkarray}
\usepackage{mathtools}
\usepackage{amsmath}
\usepackage{amssymb}
\usepackage{amsthm}
\usepackage{amsfonts}
\usepackage{mathrsfs}
\usepackage{latexsym}
\usepackage{mathtools}
\usepackage{apptools}
\usepackage{array}
\newcolumntype{C}{>{$}c<{$}}
\usepackage[switch]{lineno} 

\makeatletter
\def\moverlay{\mathpalette\mov@rlay}
\def\mov@rlay#1#2{\leavevmode\vtop{%
   \baselineskip\z@skip \lineskiplimit-\maxdimen
   \ialign{\hfil$\m@th#1##$\hfil\cr#2\crcr}}}
\newcommand{\charfusion}[3][\mathord]{
    #1{\ifx#1\mathop\vphantom{#2}\fi
        \mathpalette\mov@rlay{#2\cr#3}
      }
    \ifx#1\mathop\expandafter\displaylimits\fi}
\DeclareRobustCommand\bigop[1]{%
  \mathop{\vphantom{\sum}\mathpalette\bigop@{#1}}\slimits@
}
\newcommand{\bigop@}[2]{%
  \vcenter{%
    \sbox\z@{$#1\sum$}%
    \hbox{\resizebox{\ifx#1\displaystyle.9\fi\dimexpr\ht\z@+\dp\z@}{!}{$\m@th#2$}}%
  }%
}
\makeatother
\DeclareMathOperator{\lca}{lca}

\DeclareMathOperator*{\argmin}{arg\,min}
\newcommand{\cupdot}{\charfusion[\mathbin]{\cup}{\cdot}}

\newcommand{\unrooted}[1]{\overline{#1}}
\newtheorem{theorem}{Theorem}
\newtheorem{lemma}[theorem]{Lemma}
\newtheorem{corollary}[theorem]{Corollary}
\newtheorem{definition}[theorem]{Definition}
\newtheorem{fact}[theorem]{Observation}
\newcommand{\bmr}{\mathrel{\bm{\rightarrow}}}
\newcommand{\C}[1]{\mathscr{C}_{#1}}
\newcommand{\AX}[1]{\textnormal{#1}}

\newcommand{\child}{\mathsf{child}}
\newcommand{\parent}{\mathsf{par}}

\providecommand{\keywords}[1]{\textbf{\textit{Keywords: }} #1}

\begin{document}

\title{From Best Hits to Best Matches}

\author[1,6-12]{Peter F.\ Stadler}
\author[1]{Manuela Gei{\ss}}
\author[1]{David Schaller}
\author[2]{Alitzel L{\'o}pez S{\'a}nchez}
\author[2]{Marcos E.\ Gonz{\'a}lez}
\author[3]{Dulce I.\ Valdivia}
\author[4,5]{Marc Hellmuth}
\author[2]{Maribel Hern{\'a}ndez Rosales}

\affil[1]{Bioinformatics Group, Department of Computer Science; and 
	Interdisciplinary Center of Bioinformatics, University of Leipzig,
	H{\"a}rtelstra{\ss}e 16-18, D-04107 Leipzig, Germany}
\affil[2]{CONACYT-Instituto de Matem{\'a}ticas, UNAM Juriquilla,
	Blvd.\ Juriquilla 3001,
	76230 Juriquilla, Quer{\'e}taro, QRO, M{\'e}xico}
\affil[3]{Centro de Investigaci{\'o}n y de Estudios Avanzados del IPN, 
	CINVESTAV-Irapuato,
	Km. 9.6 Libramiento Norte Carretera Irapuato-Le{\'o}n, 
	36821 Irapuato, GTO, M{\'e}xico}
\affil[4]{Institute of Mathematics and Computer Science, 
	University of Greifswald, Walther-Rathenau-Stra{\ss}e 47, 
	D-17487 Greifswald, Germany}
\affil[5]{Center for Bioinformatics, Saarland University, Building E 2.1, 
	P.O.\ Box 151150, D-66041 Saarbr{\"u}cken, Germany}
\affil[6]{German Centre for Integrative Biodiversity Research (iDiv)
	Halle-Jena-Leipzig} 
\affil[7]{Competence Center for Scalable Data Services
	and Solutions} 
\affil[8]{Leipzig Research Center for Civilization Diseases,
	Leipzig University,
	H{\"a}rtelstra{\ss}e 16-18, D-04107 Leipzig}
\affil[9]{Max-Planck-Institute for Mathematics in the Sciences,
	Inselstra{\ss}e 22, D-04103 Leipzig}
\affil[10]{Inst.\ f.\ Theoretical Chemistry, University of Vienna,
	W{\"a}hringerstra{\ss}e 17, A-1090 Wien, Austria}
\affil[11]{Facultad de Ciencias, Universidad National de Colombia, Sede
	Bogot{\'a},
	Ciudad Universitaria, 
	111321,	Bogot{\'a}, D.C., Colombia}
\affil[12]{Santa Fe Institute, 1399 Hyde Park Rd., Santa Fe,
	NM 87501, USA}

\date{}
\normalsize

\maketitle

    \begin{abstract} 
      \noindent\textbf{Background: }           
      Many of the commonly used methods for orthology detection start from
      mutually most similar pairs of genes (reciprocal best hits) as an
      approximation for evolutionary most closely related pairs of genes
      (reciprocal best matches). This approximation of best matches by best
      hits becomes exact for ultrametric dissimilarities, i.e., under the
      Molecular Clock Hypothesis. It fails, however, whenever there are
      large lineage specific rate variations among paralogous genes.  In
      practice, this introduces a high level of noise into the input data
      for best-hit-based orthology detection methods.
      
      \smallskip
      \noindent\textbf{Results: }
      If additive distances between genes are known, then evolutionary most
      closely related pairs can be identified by considering certain
      quartets of genes provided that in each quartet the outgroup relative
      to the remaining three genes is known. \emph{A priori} knowledge of
      underlying species phylogeny greatly facilitates the identification
      of the required outgroup. Although the workflow remains a heuristic
      since the correct outgroup cannot be determined reliably in all
      cases, simulations with lineage specific biases and rate asymmetries
      show that nearly perfect results can be achieved. In a realistic
      setting, where distances data have to be estimated from sequence data
      and hence are noisy, it is still possible to obtain highly accurate
      sets of best matches.
      
      \smallskip
      \noindent\textbf{Conclusion: } 
      Improvements of tree-free orthology assessment methods can be
        expected from a combination of the accurate inference of best
        matches reported here and recent mathematical advances in the
        understanding of (reciprocal) best match graphs and orthology
        relations.
        
        \smallskip\noindent\keywords{
       	Best Matches, Gene Tree, Species Tree, Reconciliation, Orthology}
    \end{abstract}

\section{Background}

\sloppy

The distinction of orthologous and paralogous pairs of genes, respectively,
is of key importance in evolutionary biology as well as genome annotation.
As defined by Walter Fitch \cite{Fitch:70,Fitch:00}, two genes are
orthologs if their last common ancestor (in the gene tree) corresponds to a
speciation event, and they are paralogs if they arose through a duplication
event. In general, orthologs are expected to have the same function in
different organisms, while the functions of paralogs are usually similar
but clearly distinct \cite{Koonin:05,Gabaldon:13}.

A large class of computational approaches to orthology assessment
\cite{Altenhoff:09,Altenhoff:16} uses symmetric best matches (SBM)
\cite{Tatusov:97}, also known as bidirectional best hits (BBH)
\cite{Overbeek:99}, reciprocal best hits (RBH) \cite{Bork:98}, or
reciprocal smallest distance (RSD) \cite{Wall:03}. The intuitive
justification for these approaches is that symmetric best matches (in the
sense of sequence similarity) approximate the idea of evolutionarily
closest relatedness. These two concepts are not the same, however. The
notion of evolutionary relatedness depends on the underlying phylogenetic
tree $T$ and is naturally expressed by comparing last common ancestors: a gene $x$ is more closely related to a gene $y$ than to 
$y'$ if the last common ancestor $\lca(x,y)$ is a 
successor of $\lca(x,y')$ in $T$.

From an evolutionary point of view, therefore, one is interested in
\emph{reciprocal best matches} (defined in terms of the gene tree $T$)
rather than in \emph{reciprocal best hits} (defined in terms of some
distance of similarity measure between sequences). Best matches and best
hits are equivalent if and only if the Molecular Clock Hypothesis is
satisfied \cite{Zuckerkandl:62,Kumar:05}. In general this is not the
case. In particular, paralogous members of a gene family often differ in
their evolutionary rates due to (adaptive) changes in the function
\cite{Kawahara:07,Soria:14}. Both the
``Duplication-Degeneration-Complementation'' (DDC) model \cite{Force:99}
and the ``Escape from Adaptive Conflict'' (EAC) model \cite{Hittinger:07}
predict that the fate of paralogs, including their evolutionary rate, may
differ substantially between lineages that diverge soon after the
duplication event due to different selective pressures. The simplest case
is shown in Fig.\ \ref{fig:ratevar}: an ancestral gene is duplicated before
the speciation event leading to two species (indicated by colors), each
containing two paralogs (denoted by $x$ and $x'$ in the red species and $y$
and $y'$ in the blue species). The two paralogs evolve with very different
rates in the two species. Although $x$ and y as well $x'$ and $y'$ are
orthologs, the evoltionary rates are more similar between $x$ and $y'$, and
$x'$ and $y$, respectively.  This situation is not at all uncommon. The
asymmetric divergence of the genes in the HOXA cluster following the
teleost-specific (3R) genome duplication may serve as a paradigmatic
example \cite{Wagner:05a}. While in fugu (\textit{Takifugu rubripes}) and
other percomorphs the HOXAb paralogs diverge faster, it is the HOXA13b
paralog that evolves at a faster rate in zebrafish (\textit{Danio rerio}),
which diverged early from percomorphs within the Teleostei clade.

\begin{figure}[t]
	\begin{tabular}{lcr}
		\begin{minipage}{0.5\textwidth}
			\begin{center}
				\includegraphics[width=\textwidth]{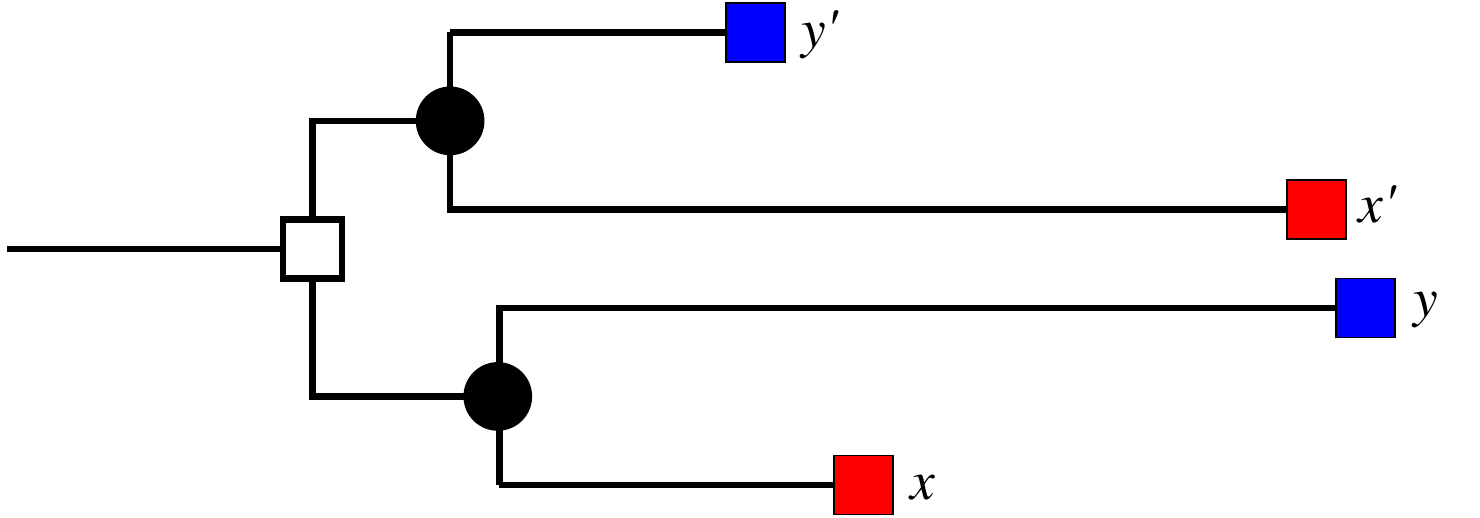}
			\end{center}
		\end{minipage}
		&    & 
		\begin{minipage}{0.4\textwidth}
			\caption{Lineage-specific rate variation between paralogs. The gene tree,
				with branch length indicating an additive evolutionary distances,
				pre-dates the speciation ($\bullet$) of the red and blue species. We
				have $\lca(x,y)\prec\lca(x,y')$ but $d(x,y')<d(x,y)$.}
			\label{fig:ratevar}
		\end{minipage}
	\end{tabular}
\end{figure}

In Fig.~\ref{fig:ratevar} the pair $x,y'$ shows the largest sequence
similarity and hence will appear as reciprocal best hit, while the closest
evolutionary relative of $x$ is the gene $y$. This discrepancy is
\textbf{not} a consequence of inaccurate measurements but an intrinsic
feature of the evolutionary process: more evolutionary events have
accumulated on the path from $x$ to $y$ than on the path from $x$ to
$y'$. The correct \emph{reciprocal best hit} therefore does not coincide
with the correct \emph{reciprocal best match}. This immediately begs the
question whether such cases can be detected from sequence comparisons.  We
consider this issue at two levels: (i) Can (reciprocal) best matches be
identified \emph{in principle}, i.e., from perfectly accurate data, and
(ii) how well can this be done in practice? To address the first question
we will assume that we can determine an additive distance between any two
genes and investigate the consequences of this assumption. To investigate
the accuracy that can be achieved from sequence data we will devise a
simulation system to generate evolutionary scenarios with complex rate
variations.

The focus on additive metrics is motivated by the close connection between
additive metrics and evolutionary trees. More precisely, an additive metric
determines a unique \emph{unrooted} phylogenetic tree $\unrooted{T}$ as
well as its branch length \cite{SimoesPereira:69,Buneman:74}, and
\emph{vice versa}. The determination of best matches, which are defined in
terms of last common ancestors, however, requires a rooted phylogenetic
tree $T$. From a theoretical point of view, therefore, the missing
information is the placement of the root of $T$ in the underlying unrooted
phylogenetic tree $\unrooted{T}$.

The problem of determining the position of the root in an unrooted tree
$\unrooted{T}$ has been well studied in the phylogenetic literature
\cite{Kinene:16}. The most common approach is the inclusion of an outgroup,
i.e., a taxon $z$ known to branch earlier than the taxa of interest. The
root is then located in the branch leading to $z$. Outgroup rooting can be
unreliable in the presence of rapid radiations or when only are very
distant outgroups are available \cite{Holland:03,Shavit:07}. The simplest
method is midpoint rooting \cite{Swofford:96}, which places the root at the
midpoint on the longest path in the tree. Despite its simplicity it often
works remarkably well \cite{Hess:07}. An interesting variation on this
theme is minimum variance rooting \cite{Mai:17}. The estimation of dated
phylogenies using a relaxed clock assumption yields an estimate for the
position of the root as a by-product \cite{Drummond:06}. A related Bayesian
method was introduced in \cite{Huelsenbeck:02}. In a phylogenomics setting,
the root of the species tree can also be obtained by minimizing the number
of inferred gene duplications \cite{Katz:12}. Most recently, non-reversible
substitution models have been employed for estimating rooted phylogenic
trees \cite{Williams:15,Cherlin:18}.

From a practical point of view, furthermore, we wish to avoid the explicit
construction of a (rooted or unrooted) gene tree $T$ since reconstructing
accurate evolutionary trees from individual gene sequences is a notoriously
difficult problem.  Instead we aim to stay as close as possible to the idea
of reciprocal best hit methods and thus we will attempt to use only
``local'' comparisons of as few as possible measurements of evolutionary
distances. This idea naturally leads us to considering quartets, i.e.,
unrooted trees describing four taxa, and the corresponding rooted
triples. It is well known that the rooted triples are sufficient to
determine the rooted tree in which they reside. Moreover, there is
polynomial-time algorithm that either constructs a rooted tree $T$ from a
set of rooted triples or determines that no such tree exists \cite{Aho:81}.
By Buneman's Theorem \cite{SimoesPereira:69,Buneman:74}, an unrooted tree
can be uniquely recovered from all its quartets. However, the problem of
determining whether a given set is compatible (i.e., whether there is an
unrooted tree $\unrooted{T}$ that contains all quartets) is NP-complete
\cite{Steel:92}, a fact that reinforces the desire to avoid the explicit
reconstruction of $\unrooted{T}$. Nevertheless, these classical results
ensure that the relevant information is contained in quartets. More
directly, we will show in this contribution that if we can reliably
determine a suitable outgroup, best matches can be extracted from a small
set of quartets.

Although much of this work is based on the assumption that an additive
distance between taxa is available, one has to keep in mind that additive
evolutionary distances, like divergence times, cannot be measured directly.
While it is common practice to determine a dissimilarity $d'(x,y)$ of two
taxa (genes) $x$ and $y$ from pairwise alignments, $d'$ is a systematic
under-estimate of the number of events $d$ due to back-mutations, and thus
not additive. In practice, the conversion of measurements of $d'$ into an
additive distance $d$ that quantifies the number of evolutionary events is
based on a Markov model of the evolutionary process. For sequence data,
this may be the Jukes-Cantor model \cite{Jukes:69} or one of its more
elaborate variants \cite{Kimura:80,Hasegawa:85,Tamura:92}. In the most
benign setting, $d$ and $d'$ are related by a monotone transformation,
which in particular implies that the measured distances $d'$ correctly
identify the best hits. It can also be shown, however, that non-additive
distances in general cannot identify the correct topology of quartets
\cite{Retzlaff:18b}.  Hence, we have no hope of computing correct best
matches from the directly measures of non-additive (dis)similarities.

This contribution is organized as follows: in the following section we give
a rigorous mathematical rendering of the background outlined above and use
it to show that, given an additive distance measure, it is indeed possible
to perfectly identify all best matches of a gene $x$ of species $s$ among
its homologs $\{y_1,\dots,y_k\}$ in species $t$ provided a suitable
outgroup $z$ can be found for every set $\{x,y_i,y_j,z\}$ of four genes. As
a consequence, the practical problem becomes the reliable identification of
correct ``relative outgroups''.  Assuming knowledge on the phylogeny of
species from which the genes are taken, we proceed to derive several
conditions under which $z$ cannot be a correct choice and use this insights
to devise a heuristic approach that works nearly perfect given additive
distance data. We then introduce (in the Methods section) a simulation
environment for generating gene family histories with complex rate
variations and show that it is possible to recover best matches more
accurately than approximating them by best hits.

\section{Theory} 

\subsection{Notation and basic definitions} 

Let $T$ be a phylogenetic (gene) tree with leaf set $L$. For each gene
$x\in L$ we denote by $\sigma(x)$ the species within which it resides. We
write $L[s]=\{y\in L\mid\sigma(y)=s\}$ for the set of genes in species
$s$. Given a rooted tree $T$ with vertex set $V$, we write $\lca(x,y)$ for
the \emph{last common ancestor} of $x,y\in V$, i.e., the vertex most
distant from the root $\rho$ that is shared by the paths from $\rho$ to $x$
and $y$, respectively. For a leaf set $L'\subseteq L$ we define the rooted
tree $T[L']$ as the tree obtained from $T$ by retaining only the vertices
and edges along paths from the root to a leaf in $L'$ and suppressing all
vertices with degree $2$.

Consider a gene $x$ in species $s$. Among all genes in species $t\ne s$,
the best matches of $x$ are all those genes $y$ in species $t$ that have
the lowest $\lca(x,y)$. These $y$ are the closest relatives of $x$ in
species $t$. This concept is made precise in
\begin{definition} \cite{Geiss:18x} 
  \label{def:genclose}
  Let $T$ be a phylogenetic tree with leaf set $L$ (denoting genes) and
  $\sigma:L\to \mathscr{S}$ identifying the species
  $\sigma(x)\in\mathscr{S}$ in which a gene $x$ resides. Then $y\in L$ is a
  \emph{best match} of $x\in L$, in symbols $x\bmr y$, if
  $\lca(x,y)\preceq \lca(x,y')$ holds for all leaves $y'$ from species
  $\sigma(y')=\sigma(y)$.
\end{definition}
The best match relation $\bmr$ is reflexive (since $\lca(x,x)=x$), but it
is neither transitive nor symmetric. Its mathematical properties are
discussed in detail in \cite{Geiss:18x,rbmg-19}. In particular, all
orthologs of $x$ are among its best matches \cite{cobmg}.

The evolutionary relatedness of two taxa $x$ and $y$ is most directly
expressed by the divergence time $\tau(x,y)$, which is the total time
elapsed in both lineages since the last common ancestor of $x$ and
$y$. Here, we consider only the case that all leaves refer to extant genes
or taxa, i.e., $\tau(x,y)=2\hat\tau(\lca(x,y))$, where $\hat\tau$ is the
age of $\lca(x,y)$.  Divergence times are ultrametric by
definition. Furthermore, there is a well-known one-to-one correspondence
between isomorphism classes of dated, rooted, phylogenetic trees and
ultrametrics, cf.\ \cite{Boecker:98,Semple:03a}. The best match relation
$\bmr$ can thus also be defined in terms of divergence time: $x\bmr y$ if
and only if
\begin{equation}
  y\in\argmin_{y'\in L[t]} \tau(x,y')
  \label{eq:argmin}
\end{equation}
The distinction between best hits and best matches thus is simply the
distance function: best matches require divergence times, while best hits
use one of several (dis)similarity measures for sequence data. They are
equivalent under the Molecular Clock Hypothesis, which however fails for
most real life data sets.

\subsection{Reconciliation of gene tree and species tree}

Since genes evolve as part of species, we can expect that \emph{a priori}
knowlege of the species phylogeny can be helpful for understanding the
phylogeny of a gene family. This link is made precise by considering the
embedding of a gene tree $T$ into a species tree $S$.

We model the species tree $S$ as a planted tree with leaf set
$\mathscr{S}$, where $0_S$ is the planted root of $S$ with its only child
$\rho_S$. The embedding of the gene tree into the species tree is
formalized by the \emph{reconciliation map} $\mu: V(T)\to V(S)\cup E(S)$,
which maps duplications to the edges of $S$ and speciations to the inner
vertices $V^0(S)$ of the species tree. We follow here the notation of
\cite{cobmg}. Restricting ourselves to duplication/loss scenarios, i.e.,
disregarding horizontal gene transfer, the reconciliation map satisfies the
root constraint \AX{(R0)} $\mu(x) = 0_S$ if and only if $x = 0_T$; the leaf
constraint \AX{(R1)} $\mu(x)=\sigma(x)$ for $x\in L(T)$, the ancestor
preservation \AX{(R2)}, i.e.,
$x\prec_{T} y \implies \mu(x)\preceq_S \mu(y)$, and the following two
speciation constraints for all speciation vertices $\mu(x)\in V^0(S)$:
\AX{(R3.i)} $\mu(x) = \lca_S(\mu(v'),\mu(v''))$ for at least two distinct
children $v',v''$ of $x$ in $T$.  \AX{(R3.ii)} $\mu(v')$ and $\mu(v'')$ are
incomparable in $S$ for any two distinct children $v'$ and $v''$ of $x$ in
$T$ \cite{cobmg}. Equivalent axiom systems are considered e.g.\ in
\cite{Doyon:11,Rusin:14,Hel-17}. Such reconciliation maps satisfy
\begin{equation}
  \mu(x)\succeq_S \lca_S(\sigma(L(T(x)))),
  \label{eq:lca}
\end{equation}
i.e, an event $x\in V(T)$ in the gene tree cannot be mapped to a node in
the species tree below the last common ancestor of all the species In this
contribution we assume that $\mu$ in addition satisfies \AX{(R4)}: If
$\mu(\lca_T(x,y))=\mu(\lca_T(x,z))\in V^0(S)$, then
$\lca_S(\sigma(x),\sigma(y))=\lca_S(\sigma(x),\sigma(z))$. In essence,
\AX{(R4)} ensures that a single node in $T$ cannot represent two distinct
speciation events, i.e., that the gene tree $T$ is not ``less resolved''
than the species tree $S$ into which it is embedded.

The reconciliation map $\mu$ defines \emph{event labels} on the inner nodes
of the gene tree $T$, identifying $u$ as a duplication node if
$\mu(u)\in E(S)$ and as speciation if $\mu(u)\in V(S)$. While it is
possibly to find a reconciliation map $\mu$ for every pair of gene and
species tree \cite{Gorecki:06}, this is no longer true when event labels on
$T$ are given \cite{HernandezRosales:12a,Hellmuth2017}. Conversely, $T$ and
$S$ imply constraints on the event labels, identifying nodes that have to
be duplications under \emph{any} reconciliation map \cite{cobmg}.  Here, we
characterize these nodes further.  We start from the following technical
results:
\begin{lemma}[{\cite[Lemma 2]{cobmg}}]
  Let $T$ be a gene tree, $S$ be a species tree and
  $\mu:V(T)\to V(S)\cup E(S)$ be a reconciliation map without horizontal
  gene transfer that does not necessarily satisfy \AX{(R4)}.  Let
  $x\in V(T)$ be a vertex with $\mu(x)\in V^0(S)$.  Then,
  $\sigma(L(T(v')))\cap \sigma(L(T(v''))) = \emptyset$ for all distinct
  $v',v''\in \child(x)$.
  \label{lem:cobmg}
\end{lemma}
Let us first consider the case of binary gene trees:
\begin{lemma}
  Let $T$ be a binary gene tree, $S$ be a species tree, and
  $\mu:V(T)\to V(S)\cup E(S)$ be a reconciliation map without horizontal
  gene transfer that does not necessarily satisfy \AX{(R4)}.  Let
  $x,y\in L(T)$ be two genes with $\sigma(x)\ne\sigma(y)$.  If
  $\lca_S(\sigma(x),\sigma(y))\prec \mu(\lca_T(x,y))$, then $\lca_T(x,y)$
  is a duplication event.
  \label{lem:dupli-bin}
\end{lemma}
\begin{proof}
  Assume for contradiction $u:=\lca_T(x,y)$ is a speciation event, i.e.,
  $\mu(u)\in V^0(S)$. Let $v'$ and $v''$ be the two children of $u$ in $T$.
  Observe that $u:=\lca_T(x,y)$ implies that $x\in L(T(v'))$ and
  $y\in L(T(v''))$ or \emph{vice versa}.  W.l.o.g.\ we assume that
  $x\in L(T(v'))$ and $y\in L(T(v''))$.  By \AX{(R3.i)} and \AX{(R3.ii)},
  $\mu(u) = \lca_S(\mu(v'),\mu(v''))$ and, in particular, $\mu(v')$ and
  $\mu(v'')$ are incomparable in $S$.  Then by Lemma \ref{lem:cobmg}, we
  have $\sigma(L(T(v')))\cap \sigma(L(T(v''))) = \emptyset$.  This and
  \AX{R2} implies that $\mu(v')\succeq_S\sigma(x)$ and
  $\mu(v')\succeq_S\sigma(y)$.  The latter two arguments imply that
  $\lca_S(\sigma(x),\sigma(y))=\mu(u)$; a contradiction.
\end{proof}

\begin{figure}[t]
	\begin{tabular}{lcr}
		\begin{minipage}{0.5\textwidth}
			\begin{center}
				\includegraphics[width=\textwidth]{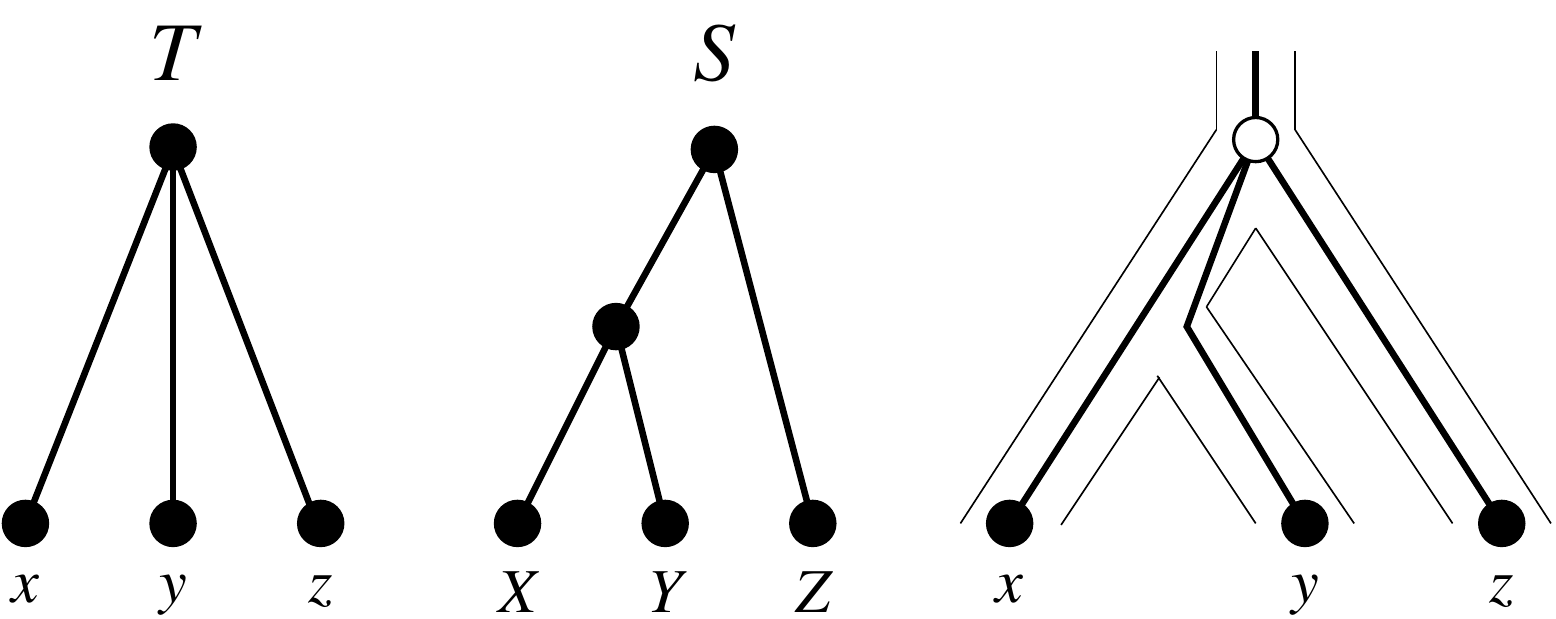}
			\end{center}
		\end{minipage}
		&    & 
		\begin{minipage}{0.4\textwidth}
			\caption{The reconciliation $\mu$ with $\mu(\lca_T(x,y,z))=\lca_S(X,Y,Z)$
				satisfied \AX{(R1)}, \AX{(R2)}, \AX{(R3.i)}, and \AX{(R3.ii)} but does
				not admit an unambiguous interpretation of $\lca_T(x,y,z)$ as single
				event: it confounds the speciation separating $Z$ and $\lca_S(Y,X)$
				with a gene duplication leading the ancestor of $x$ and $y$ or with the
				speciation separating $X$ and $Y$. In either interpretation, the
				reconciliation map $\mu$ not correspond to a mechanistic explanation of
				the gene family history.}
			\label{fig:Marc-counter}
		\end{minipage}
	\end{tabular}
\end{figure}

The assumption that $T$ is binary is necessary here as the example in
Fig.~\ref{fig:Marc-counter} shows. Such reconciliations, however, cannot be
meaningfully interpreted in terms of evolutionary events. Instead, the root
of $T$ confounds the duplication leading to $x$ and $y$ and the speciation
separating $\lca_S(\sigma(x),\sigma(y))$ from $\sigma(z)$. To suppress such
undesirable cases, we in addition require that $\mu$ satisfies axiom
\AX{(R4)}. In essence, \AX{(R4)} forbids to map two distinct speciation
events to the same vertex of $S$.
\begin{lemma}
  Let $\mu:V(T)\to V(S)\cup E(S)$ be a reconciliation map without
  horizontal gene transfer that satisfies \AX{(R4)} and let $x,y\in L(T)$
  be two genes with $\sigma(x)\ne\sigma(y)$.  If
  $\lca_S(\sigma(x),\sigma(y))\prec \mu(\lca_T(x,y))$, then $\lca_T(x,y)$ is
  a duplication event.
  \label{lem:dupli}
\end{lemma}
\begin{proof}
  We assume that $T$ is non-binary since the binary case is covered already
  by Lemma \ref{lem:dupli-bin}. Moreover, we assume, for contradiction,
  that $u:=\lca_T(x,y)$ is a speciation event, i.e., $\mu(u)\in V^0(S)$.
  Let $v_x$ and $v_y$ be the children of $u$ with $x\preceq_T v_x$ and
  $y\preceq_T v_y$; thus we have $\sigma(x)\in \sigma(L(T(v_x)))$ and
  $\sigma(y)\in \sigma(L(T(v_y)))$.  Since $u=\lca_T(x,y)$, $v_x$ and $v_y$
  are incomparable in $T$ and hence $v_x\neq v_y$. By \AX{(R3.i)},
  $\mu(v_x)$ and $\mu(v_y)$ are incomparable in $S$.  Lemma \ref{lem:cobmg}
  implies $\sigma(L(T(v')))\cap \sigma(L(T(v''))) = \emptyset$ for all
  distinct children $v'$ and $v''$ of $u$. The latter two fact together
  with \AX{(R2)} imply
  $\lca_S(\sigma(x),\sigma(y)) = \lca_S(\mu(v_x),\mu(v_y))\prec \mu(u)$.
  By \AX{(R3.i)}, $\mu(u) = \lca_S(\mu(v'),\mu(v''))$ for some children
  $v'$ and $v''$ of $u$, and thus
  $\lca_S(\mu(v'),\mu(v'')) = \lca_S(\sigma(z'),\sigma(z''))$ for some
  leaves $z'\in L(T(v'))$ and $z''\in L(T(v''))$ from different species
  $\sigma(z')\ne\sigma(z'')$.

  We proceed by showing that for at least one of $\sigma(z')$ and
  $\sigma(z'')$ we have
  $\lca_S(\sigma(x),\sigma(z')) = \lca_S(\sigma(z'),\sigma(z''))$ or
  $\lca_S(\sigma(x),\sigma(z'')) = \lca_S(\sigma(z'),\sigma(z''))$.
  Suppose that
  $\lca_S(\sigma(x),\sigma(z')) \neq \lca_S(\sigma(z'),\sigma(z''))$.
  Hence,
  $\lca_S(\sigma(x),\sigma(z')) \prec_S \lca_S(\sigma(z'),\sigma(z'')) =
  \mu(u)$.  Therefore,
  $\lca_S(\sigma(x),\sigma(z'')) = \lca_S(\sigma(z'),\sigma(z''))$.
  Similarily, if
  $\lca_S(\sigma(x),\sigma(z'')) \neq \lca_S(\sigma(z'),\sigma(z''))$, then
  $\lca_S(\sigma(x),\sigma(z')) = \lca_S(\sigma(z'),\sigma(z''))$.
  Hence, assume w.l.o.g.\ that 
  $\lca_S(\sigma(x),\sigma(z')) = \lca_S(\sigma(z'),\sigma(z'')) \neq
  \lca_S(\sigma(x),\sigma(y))$
  Now, by contraposition of \AX{(R4)}, we have 
  $\mu(u)  = \mu(\lca_T(x,y))\neq \mu(\lca_T(x,z')) = \mu(u)$; a contradiction.
\end{proof}

Lemma \ref{lem:dupli} conveniently generalizes to sets of genes:
\begin{corollary}
  Let $\mu:V(T)\to V(S)\cup E(S)$ be a reconciliation map without
  horizontal gene transfer that satisfies \AX{(R4)} and let
  $A\subseteq L(T)$ with $|\sigma(A)|\ge 2$. If
  $\lca_S(\sigma(A))\prec \mu(\lca_T(A))$, then $\lca_T(A)$ is a duplication
  event.
  \label{cor:dupli}
\end{corollary}
\begin{proof}
  Note first that $\lca_T(A) = \lca_T(x,y)$ for some $x,y\in A$.  Assume
  first $\sigma(x)\neq \sigma(y)$.  Thus, 
  $\lca_S(\sigma(A))\prec \mu(\lca_T(A))$ implies
  $\lca_S(\sigma(x),\sigma(y)) \preceq \lca_S(\sigma(A)) \prec
  \mu(\lca_T(A))=\mu(\lca_T(x,y))$.  Hence, the statement follows from
  Lemma \ref{lem:dupli}. If $\sigma(x)=\sigma(y)$, then
  $\lca_T(A) = \lca_T(x,y)$ implies that there are distinct children $v_x$
  and $v_y$ of $\lca_T(A)$ with $v_x\succeq x$ and $v_y\succeq y$. Thus,
  $\lca_T(A) = \lca_T(v_x,v_y)$.  However, since $\sigma(x)=\sigma(y)$ we
  have $\sigma(L(T(v_x)))\cap \sigma(L(T(v_y)))\neq \emptyset$.  Thus,
  Lemma \ref{lem:cobmg} implies that $\mu(\lca_T(A))\notin V^0(S)$ and
  hence, $\lca_T(A)$ is duplication.
\end{proof}

\subsection{Trees and (dis)similarities} 

Neither the divergence times nor the $\lca$ function of the phylogenetic
tree $T$ can be measured directly. The next-best choice is to work with an
evolutionary distance, which measures the number of evolutionary events
that have taken place to separate two taxa. For each edge $e=uv$ in $T$ it
is given by $\ell(e)=\int_{\hat\tau(u)}^{\hat\tau(v)}\mu_e(t)dt$, where
$\mu_e(t)$ is the rate of evolution. In general $\mu_e(t)$ depends both on
the lineage, and thus the individual edges in $T$, as well as on the exact
point in time along $e$. It associates with each edge $e$ a measure
$\ell(e)$ of changes incurred, and thus an additive distance. If
$\mu_e(t)=\mu_0$ is constant, we simply have
$d_{\ell,T}(x,y)=\mu_0 \tau(x,y)$. This is the well-known Molecular Clock
Hypothesis \cite{Zuckerkandl:62,Kumar:05}.

In general, we consider $\ell: E(T)\to\mathbb{R}^+$ simply as an assignment
of positive lengths to the edges of $T$, which we interpret as a measure
proportional to the number of evolutionary events. This gives rise to a
metric distance function $d_{T,\ell}(x,y)$ on $L$ defined as the sum of the
lengths $\ell(e)$ of the edges $e$ along the unique path connecting $x$ and
$y$ in $T$. From $T$ we obtain an associated \emph{unrooted} tree
$\unrooted{T}$ by (i) omitting the planted root $0_T$ and its incident
edge, and (ii), in case the root $\rho$ in $T$ has exactly two children
$u_1$ and $u_2$, by replacing the path $u_1\rho u_2$ by a single edge
$u_1u_2$ with length $\ell(u_1u_2):=\ell(u_1\rho)+\ell(\rho u_2)$. Note
that the dissimilarity function $\ell$ is by construction the same on $T$
and $\unrooted{T}$. Thus $\unrooted{T}$ determines $T$ up to the position
of the root, i.e., $T$ is obtained from $\unrooted{T}$ by inserting the
root into an edge of $\unrooted{T}$ or declaring an inner vertex of
$\unrooted{T}$ as the root.  As for rooted trees, we define the restriction
$\unrooted{T}[L']$ for some subset $L'\subseteq L$ by retaining only the
vertices and edges along the paths between pairs of vertices in $L'$ and
then suppressing all vertices of degree $2$. We note that
$\unrooted{T[L']}=\unrooted{T}[L']$.

A dissimilarity $d$ on $L$ is called \emph{additive} if there is an
unrooted tree $\unrooted{T}$ with edge lengths $\ell$ such that
$d=d_{\ell,\unrooted{T}}$. A key result in mathematical phylogenetics
\cite{SimoesPereira:69,Buneman:74} characterizes additive (pseudo)metrics
as those that satisfy the \emph{four point condition}. It states that $d$
is additive if and only if the restriction of $d$ to each subset $L'$ of
$L$ with $|L'|=4$, usually called a \emph{quartet}, is additive and thus
determines a tree on four leaves. Furthermore, the unrooted tree
$\unrooted{T}$ is uniquely defined by $d$. In principle, therefore,
distance data completely determines a phylogenetic tree up to the position
of the root.

The results of \cite{SimoesPereira:69,Buneman:74} furthermore imply that
$\overline{T}$ can be expressed in terms of its four-taxa subtrees.  This
provides us with a natural possibility to consider only ``local''
topologies instead of the having to construct the unrooted tree
$\unrooted{T}$ explicitly. To this end, we consider the restrictions
$\unrooted{T}[p,q,r,s]$ of $\unrooted{T}$ to four distinct leaves
$p,q,r,s\in L$ and define the \emph{quartet relation}
\cite{Sattah:77,Fitch:81} $(pq|rs)$ if there is an edge $e$ in
$\unrooted{T}$, and thus in $\unrooted{T}[p,q,r,s]$, such that $\{p,q\}$
and $\{r,s\}$ are in different connected components of the forest obtained
by removing $e$ from $\unrooted{T}$ or $\unrooted{T}[p,q,r,s]$.
Equivalently, we have \cite{Sattah:77,Fitch:81}
\begin{equation}
  \begin{split}
    &(pq|rs) \iff \\
    &d(p,q)+d(r,s) < d(p,r)+d(q,s), d(p,s)+d(q,r)\,.
  \end{split}
  \label{eq:quartet}
\end{equation}
In fact, for additive metrics, the two distance sums on the r.h.s.\ are
equal \cite{SimoesPereira:69,Buneman:74}. All three terms are equal if and
only if the four points form a star, whence the existence of a separating
edge requires the strict inequality. 
By a slight abuse of notation we write
$\unrooted{T}[p,q,r,s]=(pq|rs)$ if Equ.\ (\ref{eq:quartet}) holds, and
$\unrooted{T}[p,q,r,s]=\bm{\times}$ if no quartet exists on these four
leaves, i.e., if $\unrooted{T}[p,q,r,s]$ is the star tree. 

\subsection{From Quartets to rooted triples}
In a \emph{planted phylogenetic tree} $T$ with leaf set $L\cup\{0_T\}$ all
inner vertices have degree at least $3$. The special leaf $0_T$ identifies
the ancestral state. Its only neighbor is the root $\rho_T$. The vertex set
of $T$ is endowed with a partial order $\prec$ such that $x\preceq y$
whenever $y$ lies along the unique path connecting $x$ and the planted root
$0_T$ and thus also the root $\rho_T$. Thus the leaves are the minimal
elements w.r.t.\ $\prec$. Furthermore, for $A \subseteq L$ we define
$\lca(A)=\min\{z\mid x\preceq z \text{ for all } x\in A\}$,
where the minimum is taken w.r.t.\ the partial order $\prec$.
For every $u\in V(T)$, we denote by $T(u)$ the
subtree of $T$ rooted in $u$. It will sometimes be useful to consider
$T(u)$ as planted tree by including the unique parent $v$ of $u$ and the
edge $vu$. The leaf set of $T(u)$ will be denoted by $L(T(u))$.

The most common method to specify the root of a phylogenetic tree is the
use of so-called outgroups, that is, additional taxa that are known \emph{a
  priori} to be outside a monophyletic group of interest. Given a planted
(or rooted) phylogenetic tree, on the other hand, monophyletic groups are
the leaf sets of a subtree, i.e., $L'$ is a monophyletic group if and only
if there is a vertex $u\in V(T)$ such that $L'=L(T(u))$. Every leaf
$x\in L\setminus L'$ is an outgroup for $L'$.

Every edge in an unrooted tree $\unrooted{T}$ defines a split $L'|L''$ of
$L$, where $L'$ and $L''$ are the leaves in the connected components of
$\unrooted{T}\setminus e=\unrooted{T'}\cupdot\unrooted{T''}$.  At most one
of the two subtrees $\unrooted{T'}$ and $\unrooted{T''}$ contains the root
of the underlying phylogenetic tree $T$.  If the root is not contained in
$\unrooted{T'}$, then the tree $T'\cup\{e\}$ planted at the endpoint of $e$
describes a monophyletic group. In this case all $x\in L''$ are outgroups
for $T'$. Which subtrees of $\unrooted{T}$ correspond to monophyletic
groups is determined by the position of the root, and therefore requires
external information.

\begin{figure}
  \begin{center}
    \includegraphics[width=0.8\columnwidth]{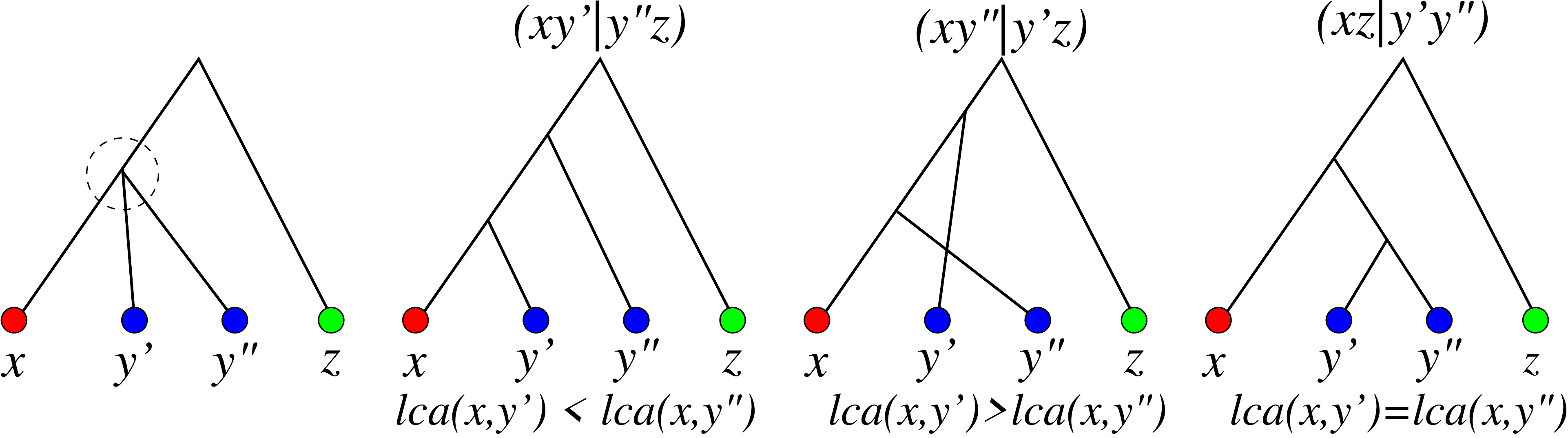}
  \end{center}
  \caption{Relation of last common ancestors $\lca(x,y')$ and
    $\lca(x,y'')$, resp., with quadruples on $\{x,y',y'',z\}$ with a
    trusted outgroup $z$.}
  \label{fig:wtf}
\end{figure}

It will be convenient in the following to define outgroups not only for
monophyletic groups.
\begin{definition} 
  For a phylogenetic tree $T$ with leaf set $L$, consider a subset
  $L'\subseteq L$ and a leaf $z\in L\setminus L'$. We say that $z$ is an
  \emph{outgroup} for $L'$ if $\lca(L')\prec \lca(L',z)$. 
\end{definition}

Let us now return to the quartets of $\unrooted{T}$. The following simple
result, illustrated in Fig.\ \ref{fig:wtf}, shows that quartets can be used to
infer inequalities between $\lca$ vertices in $T$ provided one of the four
leafs is known to be an outgroup for the other three:
\begin{lemma}
  Suppose $z$ is an outgroup for $\{x,y',y''\}$ in $T$. If
  $\unrooted{T}[x,y',y'',z]$ is fully resolved, then
  \begin{itemize}
  \item[(i)] $\lca(x,y')=\lca(x,y'')$ iff 
    $\unrooted{T}[x,y',y'',z]=(xz|y'y'')$,
  \item[(ii)] $\lca(x,y')\prec\lca(x,y'')$ iff 
    $\unrooted{T}[x,y',y'',z]=(xy'|y''z)$, \\and 
  \item[(iii)] $\lca(x,y')\succ\lca(x,y'')$ iff
    $\unrooted{T}[x,y',y'',z]=(xy''|y'z)$.
  \end{itemize}
  Otherwise,
  $\unrooted{T}[x,y',y'',z]=\bm{\times}$ and $\lca(x,y')=\lca(x,y'')$.
\label{lem:ineq}
\end{lemma}
\begin{proof}
  Since $z$ is an outgroup by assumption, there are only three possible
  fully resolved rooted tree with $L=\{x,y',y'',z\}$, see
  Fig.~\ref{fig:wtf}.  Each of these trees corresponds to a unique
  quadruple (annotated at the top). The relationship between $\lca(x,y')$
  and $\lca(x,y'')$ is determined by the tree topology. The statement
  follows by inspecting the three cases. If $\bar T[x,y',y'',z]$ is not
  fully resolved, no quartet is defined on $\{x,y',y'',z\}$, i.e., $\bar T$
  is the star tree and thus $\lca(x,y')=\lca(x,y'')=\lca(y',y'')$.
\end{proof}

\begin{fact}
  If $u'=\lca(x,y')$ and $v'=\lca(x,y'')$ for $x,y',y''\in L$, then $u'$
  and $v'$ are comparable w.r.t.\ $\preceq$ in $T$.
  \label{fact:lca}
\end{fact}

Lemma~\ref{lem:ineq} together with Obs.\ \ref{fact:lca} implies that
quadruples with known outgroups can be used to identify best matches. More
precisely, in order to determine the set $\{y\in L[s] \mid x\bmr y\}$ it
suffices to consider leaf sets $\{x,y',y'',z\}$ with $y',y''\in L[s]$ such
that $z$ is an outgroup for $\{x,y',y''\}$. By Lemma~\ref{lem:ineq}, any
set of this type implies an (in)equality between $\lca(x,y')$ and
$\lca(x,y'')$. It may not be necessary to consider all quadruples. To
explore ways to reduce the computational efforts, let us assume that for
given $x\in L$ and $s\in S$, $s\ne \sigma(x)$, we can identify sets
$Y\subseteq L[s]$ and $Z\subseteq L$ such that the following three
assumptions are satisfied: 
\begin{description}
\item[\AX{(A0)}] The noise in the data is small enough so that for any four
  taxa $\{x,y',y'',z\}$ with $y',y''\in Y$ and $z\in Z$ one of the three
  possible quartets or the star topology is inferred correctly.
\item[\AX{(A1)}] The candidate set $Y\subseteq L[s]$ contains all best
  matches of $x$ in species $s$ (but usually also additional leaves).
\item[\AX{(A2)}] $Z$ is a non-empty set outgroups for $Y\cup\{x\}$.
\end{description}

Before we proceed, let us consider these three assumptions in some more
detail.  \AX{(A0)} is satisfied by construction for additive distance data.
In real-life applications it is often possible to obtain at least a very
good approximation using explicit models of sequence evolution. In
addition, several computational approaches have been proposed to estimate
the quartet relation directly from sequence data. It is also worth noting
that \AX{(A0)} does not require precise distance data, it only asks for
correct categorical data on the quartet relation.

Condition \AX{(A1)} can always be enforced by setting $Y=L[s]$. We make
this assumption explicit because in practice it will be desirable to work
with small subsets $Y\subseteq L[s]$ as using $L[s]$ may be too expensive
for large gene families. The inclusion of very distant relatives may be
problematic for the construction of good multiple sequence alignments and
thus the extraction of the quartet relation. Furthermore, it may be
difficult to find suitable outgroup data in this case. Thus we will limit
$Y$ to a manageable size and sufficient sequence similarity. In
\texttt{ProteinOrtho} \cite{Lechner:11a}, for example, $Y\subseteq L[s]$ is
defined as the set sequence with \texttt{blast} bit-scores exceeding a
certain fraction of the best hit for $x$ in species $s$.

Condition \AX{(A2)}, i.e., the knowledge of appropriate outgroups, is the
only problematic assumption. As discussed above, distance-based methods by
construction do not convey information on the root of the phylogenetic tree
$T$ but only determine its unrooted version $\unrooted{T}$. As a
consequence, additional information, not contained in the pairwise distance
measurements, is necessary to determine the edge in $\unrooted{T}$ that
harbors the position of the root $\rho$ of $T$ \cite{Penny:76}. In general,
$Z$ will be chosen from one or more species that are outgroups to
$\sigma(x)$ and $s$ in $S$. Even if outgroup species are given, gene
duplications may pre-date the divergence of the available species set, so
that a given data set will usually violate \AX{(A2)} for some pairs of
leaves. We will return to these issues in more detail in the following
sections.

\begin{algorithm}
  \caption{Overall Workflow}
  \label{alg:overall}
  \begin{algorithmic}[1]
    \REQUIRE reference vertex $x$
    \STATE retrieve a sufficient set $Y\subseteq L[s]$ of candidate
           best matches for $x$ with color $s$
    \STATE determine a set $Z$ of outgroup vertices for $Y\cup\{x\}$
    \STATE initialize an edgeless digraph $\Gamma$ with vertex set $Y$
    \FORALL{pairs $y',y''\in Y$} \label{alg-start1}
       \FORALL{$z\in Z$}
          \STATE determine significantly supported quartet on
              $\{x,y',y'',z\}$
       \ENDFOR
       \STATE determine consensus quartet over all choices of $z\in Z$
       \label{alg-end1}
       \IF {consensus quartet implies $\lca(x,y_1)\preceq\lca(x,y_2)$}
          \STATE insert the directed edge $(y_2,y_1)$ into $\Gamma$
       \ENDIF  
    \ENDFOR
    \STATE compute the strongly connected components of $\Gamma$
    \STATE report strongly connected components without out-edges as
       the set of best matches $\{y\in Y| x\bmr y\}$
  \end{algorithmic}
\end{algorithm}

The discussion so far suggests to use the quadruple relation for sets of
the form $\{x,y',y',z\}$ with $y',y''\in Y$ and $z\in Z$ to determine the
best matches of $x$ in the species containing the homolog set $Y$. The
procedure is summarized in Alg.\ \ref{alg:overall}. The main result of this
section establishes its correctness.

\begin{theorem}
  Algorithm~\ref{alg:overall} correctly identifies the set of best matches
  of $x$ with color $s$ as the unique strongly connected component of
  $\Gamma$ without out-edges provided assumptions \AX{(A0)}, \AX{(A1)}, and
  \AX{(A2)} are satisfied.
\label{thm:erewan}
\end{theorem}
\begin{proof}
  Assumptions \AX{(A1)} and \AX{(A2)} imply that comparison of the last
  common ancestors can be performed in terms of the quartets according to
  Lemma~\ref{lem:ineq}, which by assumption \AX{(A0)} are all inferred
  correctly. Therefore, lines \ref{alg-start1}-\ref{alg-end1} compute all
  quartets correctly, and thus the inequality between $\lca(x,y_1)$ and
  $\lca(x,y_2)$ is inferred correctly. The auxiliary graphs $\Gamma$
  therefore contains at least one arc between any two vertices
  $y',y''\in Y$ and both the arc $(y',y'')$ and $(y'',y')$ if and only if
  $\lca(x,y')=\lca(x,y'')$, i.e., the strongly connected components are
  cliques. Since the $\lca(x,y)$ are interior vertices of $T$ that are
  totally ordered along the path from $x$ to the root of $T$ (Observation
  \ref{fact:lca}), there is a unique strongly connected component $B$ in
  $\Gamma$ that has no out-edges, whose vertices are those $y\in B$ for
  which $\lca(x,y)$ is minimal. Thus $B$ is the set of best matches of $x$
  with color $s$.
\end{proof}
Algorithm~\ref{alg:overall} therefore works correctly at least under
idealized assumptions. It also serves as a heuristic in cases where one of
the assumptions (usually \AX{(A2)}) is violated.

\subsection{Identification of outgroups}

In many practical applications, the phylogenetic relationships between the
\emph{species} under consideration are known. We therefore investigate here
to what extent knowledge of the species tree $S$ can help to identify good
outgroup sets $Z$.  Ideally, the genes chosen as outgroups $Z$ are
co-orthologs of the focal gene set $Y$, i.e., the duplication event that
produced $y'$ and $y''$ occured after the speciation event that separates
$\sigma(z)$ for all $z\in Z$ from $\sigma(X)$ and $\sigma(Y)$. As we
shall see, it is not possible to identify outgroups with complete
certainty. It is possible, however, to identify incorrect choices in many
situations.

In the following we consider three species $\sigma(X)$, $\sigma(Y)$, and
$\sigma(z)$ for $z\in Z$ such that
\begin{equation}
  \lca_S(\sigma(X),\sigma(Y)) \prec_S
  \lca_S(\sigma(X),\sigma(Y),\sigma(z)),
\end{equation}
i.e., $\sigma(z)$ is an outgroup in the species tree for $\sigma(X)$ and
$\sigma(Y)$. Problematic cases in which quartets are interpreted
incorrectly may appear whenever the duplication event $\lca_T(y',y'')$
separating two paralogs $y',y''\in Y$ pre-dates the speciation event
separating $\sigma(z)$ from $\lca_S(\sigma(X),\sigma(Y))$. We capture this
situation in
\begin{definition}
  Let $u$ be an inner node of the species tree $S$, let $y',y''\in Y$ be
  paralogs in a species $\sigma(Y)\in L(S(u))$. Then $\lca_T(y',y'')$ is an
  \emph{ancient duplication relative to $u\in V(S)$} for the reconciliation
  map $\mu: V(T)\to V(S)\cup E(T)$ if $u \prec_S \mu(\lca_T(y',y''))$.
\end{definition}
Clearly, if $\lca_T(y',y'')$ is an ancient duplication relative to
$\lca_S(\sigma(X),\sigma(Y),\sigma(z))$, then genes in $z\in Z$ are bad
choices as outgroups $\{x,y',y'',z\}$.  The difficulty is that we do not
know the reconciliation map $\mu$ in our setting. In some cases, however,
it is possible to identify vertices in $T$ that are ancient duplications
relative to some speciation for \emph{any} reconciliation. Such cases can
then be avoided.

\begin{figure}
  \begin{center}
    \includegraphics[width=0.6\columnwidth]{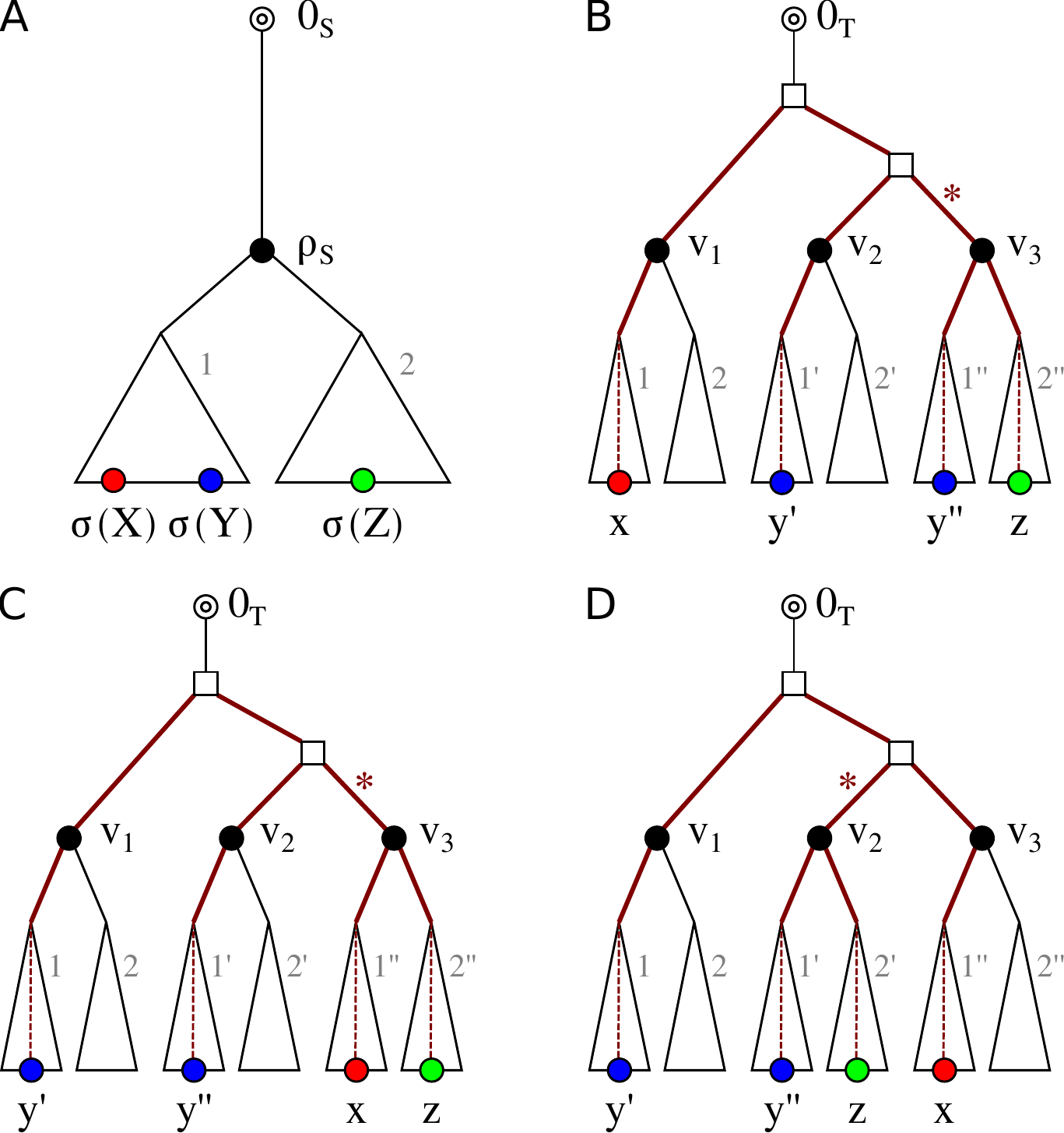}
  \end{center}
  \caption{Minimal examples in which ancient duplications lead to false
    positives (FP) or false negatives (FN) when choosing outgroups as
    described in the text. (A) The species tree $S$ displaying the triple
    ($\sigma(X)\sigma(Y)|\sigma(Z)$).  (B-D) Gene trees $T$ with two
    ancient duplications. We assume that $y'$ and $y''$ are the only extant
    genes of color $\sigma(Y)$, i.e. color $\sigma(Y)$ is extinct in the
    subtree of $x$ in each of the shown cases. The asterisk marks the
    discriminating edge for the quartet inference.  (B) A quartet
    $(xy'|y''z)$ is inferred so that only $y'$ but not $y''$ is a best
    match, $(x, y'')$ is a FN. (C) A quartet $(xz|y'y'')$ is inferred so
    that $y'$ is a false best match, $(x, y')$ is a FP. (D) A quartet
    $(xy'|y''z)$ is inferred so that $(x, y')$ is a FP and $(x, y'')$ is a
    FN.}
  \label{fig:ancient-dupl}
\end{figure}

Before we investigate possibilities to identify some ancient duplications
in distance data, we prove a rather technical result that shows that in
cases without too many ancient duplications, Algorithm \ref{alg:overall}
produces correct results. For the proof we will need to consider the
reconciliation map $\mu$ for \emph{complete} gene family histories, i.e.,
gene trees $T$ containing extant genes as well as all branches leading to
loss events. As above, we do not consider HGT. The leaf set of $T$ is thus
$L:= L_e \cupdot L_0$, where $L_e$ represents the extant genes and $L_0$
denotes loss events. Since the species map is naturally restricted to
extant genes (i.e., $\sigma: L_e(T)\to L(S)$), we need to restrict
\AX{(R1)}: If $x\in L_e(T)$, then $\mu(x)=\sigma(x)$. We will refer to such
gene trees and reconciliation maps as \emph{extended} gene trees and
\emph{extended} reconciliation maps, respectively. Correspondingly, Lemma
\ref{lem:cobmg} only holds for $L_e$, i.e., we can conclude that
$\sigma(L_e(T(w_1)))\cap \sigma(L_e(T(w_2)))=\emptyset$.  This can easily
be seen by reusing the contradiction argument in \cite{cobmg}[Lemma 2].  As
a consequence of loss events we now may have $\sigma(L_e(T(v)))=\emptyset$
for some nodes $v\in V(T)$.

\begin{lemma}
  Let $(T,\sigma)$ be an extended gene tree with a non-empty set of extant
  genes $L_e=X\cupdot Y\cupdot Z$ with $|\sigma(Z)|=1$, let $S$ be a
  species tree on $S=\{\sigma(X)$, $\sigma(Y)$, $\sigma(Z)\}$ such that
  $\lca_S(\sigma(X),\sigma(Y))\prec
  \lca_S(\sigma(X),\sigma(Y),\sigma(Z))=\rho_S$, and let $\mu$ be an
  extended reconciliation map from $(T,\sigma)$ to $S$. If \AX{(A0)} holds
  and $|\mu^{-1}(\rho_S)|\le 2$, then Algorithm \ref{alg:overall}, using
  $Y$ as the candidate best match set and $Z$ as outgroup set, correctly
  determines, for every gene $x\in X$, all best matches in species
  $\sigma(Y)$.
  \label{lem:genes_in_rhoS}
\end{lemma}
\begin{proof}
  First note that the statement is trivial if there exists only one gene in
  $Y$. Hence, we can assume that $Y$ contains more than one gene.
  Moreover, Condition \AX{(A1)} is trivially satisfied since, by
  assumption, the candidate set of best matches of $x$ in $Y$ is exactly
  $Y$. Since $L_e$ is non-empty, we have to consider the two cases
  $|\mu^{-1}(\rho_S)|=1$ and $|\mu^{-1}(\rho_S)|=2$.

  Assume first $|\mu^{-1}(\rho_S)|=1$, i.e., there exists exactly one
  $v\in V(T)$ such that $\mu(v)=\rho_S$. We then have
  $\sigma(L_e(T(w_1)))\cap \sigma(L_e(T(w_2)))=\emptyset$ for any distinct
  $w_1,w_2\in\child_T(v)$ (cf.\ \cite{cobmg}[Lemma 2], Lemma
  \ref{lem:cobmg}), which, by construction of the species tree $S$,
  immediately implies
  $\sigma(L_e(T(w)))\in \{\{\sigma(X),\sigma(Y)\},\{\sigma(Z)\}\}$ for any
  $w\in\child_T(v)$. Hence, $Z$ is an outgroup set for $Y\cup \{x\}$, i.e.,
  Condition \AX{(A2)} is satisfied, and the statement thus follows directly
  from Theorem \ref{thm:erewan}.

  Now suppose $|\mu^{-1}(\rho_S)|=2$, i.e., there are exactly two distinct
  $v_1,v_2\in V(T)$ with $\mu(v_1)=\mu(v_2)=\rho_S$. Let
  $T_1\coloneqq T(v_1)$ and $T_2\coloneqq T(v_2)$ be the subtrees of $T$
  rooted at $v_1$ and $v_2$, resp., and assume w.l.o.g.\ $x\in
  L_e(T_1)$. Note that $L_e(T_1)\cup L_e(T_2)=L_e$.  Let
  $w_1\in \child_T(v_1)$ such that $x\preceq_T w_1 \prec_T v_1$. If $w_1$
  were mapped to an edge or vertex along the path from $\rho_S$ to
  $\sigma(Z)$, then
  $\lca_S(\sigma(X),\sigma(Y))\prec
  \lca_S(\sigma(X),\sigma(Y),\sigma(Z))=\rho_S$ would imply
  $\sigma(X)\not\preceq_S \mu(w_1)$; a contradiction to \AX{(R2)}.  Thus,
  $\sigma(Z)\notin \sigma(L_e(T(w_1)))$.  Since $\mu(v_1)\in V^0(S)$,
  Condition \AX{(R3.i)} implies that there exists $w_2\in\child_T(v_1)$,
  $w_2\neq w_1$, such that $\mu(v_1)=\lca(\mu(w_1),\mu(w_2))$.  Since
  $\sigma(L_e(T(w_1)))\cap \sigma(L_e(T(w_2)))=\emptyset$ by Lemma
  \ref{lem:cobmg}, we obtain
  $\sigma(L_e(T(w_1)))\subseteq\{\sigma(X),\sigma(Y)\}$ and
  $\sigma(L_e(T(w_2)))\subseteq\{\sigma(Z)\}$.  We distinguish the two
  cases (a) $\sigma(Y)\notin \sigma(L_e(T_1))$ and (b)
  $\sigma(Y)\in \sigma(L_e(T_1))$.

  \noindent\textit{Case (a):} If $\sigma(Y)\notin \sigma(L_e(T_1))$, any
  leaf $y\in Y$ must reside within a subtree $T(w')$ with
  $w'\in\child_T(v_2)$, thus all genes in $Y$ are best matches of
  $x$. Since the speciation node $v_2$ separates $\sigma(Z)$ from
  $\sigma(X)$ and $\sigma(Y)$, we have $\sigma(Z)\notin \sigma(L_e(T(w')))$
  for any such $w'$ (cf.\ Lemma \ref{lem:cobmg}). Moreover, reusing the
  same arguments as for $v_1$, we conclude that there exists exactly one
  such $w'\in\child_T(v_2)$ such that
  $\sigma(Y)\in\sigma(L_e(T(w')))$. Hence, any two distinct $y,y'\in Y$
  reside within the same subtree $T(w')$ and thus
  $\lca_T(x,y)=\lca_T(x,y')$.  Since $\sigma(Z)\notin \sigma(L_e(T(w')))$,
  this immediately implies $\unrooted{T}[x,y,y',z]=(xz|yy')$ for any
  $z\in Z$. Hence, $\Gamma$ is the complete graph, i.e., any gene of
  species $\sigma(Y)$ is correctly inferred as a best match of $x$.
  
  \noindent\textit{Case (b):} Assume, for contradiction, that there exists
  $w_3\in\child_T(v_1)\setminus\{w_1\}$ such that
  $\sigma(Y)\in\sigma(L_e(T(w_3)))$. Clearly, $w_3\neq w_2$. Since it must
  hold $\sigma(L_e(T(w_1)))\cap\sigma(L_e(T(w_3)))=\emptyset$ as well as
  $\sigma(L_e(T(w_2)))\cap\sigma(L_e(T(w_3)))=\emptyset$ by Lemma
  \ref{lem:cobmg}, we conclude $\sigma(L_e(T(w_1)))=\{\sigma(X)\}$ and
  $\sigma(L_e(T(w_3)))=\{\sigma(Y)\}$. However, \AX{(R4)} then implies
  $\lca_S(\sigma(X),\sigma(Y))=\lca_S(\sigma(Y),\sigma(Z))$; a
  contradiction. Hence, there exists an extant gene $y\preceq_T w_1$ in
  $Y$. Then, as $\sigma(Z)\notin \sigma(L_e(T(w_1)))$, any $z\in Z$ infers
  the same quartet on $\{x,y,y',z\}$, $y'\in Y\setminus \{y\}$. We
  therefore conclude that the auxiliary graph $\Gamma$ contains a unique
  strongly connected component without out-edges, which represents the set
  of best matches of $x$ in $Y$. Note that in these cases Condition
  \AX{(A2)} is not necessarily satisfied, but Algorithm \ref{alg:overall}
  still provides the exact solution.
\end{proof}
The condition $|\mu^{-1}(\rho_S)|\le 2$ makes an explicit assumption on the
true history of the gene family by limiting the scenario to at most one
ancient duplication on $X\cupdot Y\cupdot Z$. Fig.\ \ref{fig:ancient-dupl}
shows that this condition cannot be dropped: if there are two or more
ancient duplications affecting $X$, $Y$, and $Z$, then the correct
inference of best matches from quartets can no longer be guaranteed.  It is
important to note that the condition $|\mu^{-1}(\rho_S)|\le 2$ cannot be
checked in real data since $\mu$ is unknown. In the simulated data,
however, it is easy to validate and we observed empirically that it is
rarely violated in our data (see Simulation Results section).

In some situations ancient duplications can be inferred unambiguously,
independent of the reconciliation map $\mu$. This is in particular the case
if there are incongruences between quartets of genes and species.  Consider
four genes $a,b,c,d$ residing in four pairwise distinct species
$\sigma(a)$, $\sigma(b)$, $\sigma(c)$, and $\sigma(d)$, and assume that
these four species form the quartet
$(\sigma(a)\sigma(b)|\sigma(c)\sigma(d))$.  Then we say that the gene and
species quartets are \emph{congruent} if $\unrooted{T}[a,b,c,d]=(ab|cd)$ or
$\times$. Otherwise, i.e., for
$\unrooted{T}[a,b,c,d]\in\{ (ac|bd), (ad|bc)\}$, we say they are
\emph{incongruent}, see Fig.\ \ref{fig:ancient}.  In the following we show
that the incogruence of gene and species quartets implies ancient
duplications. More precisely:

\begin{figure}
  \begin{center}
    \includegraphics[width=0.6\columnwidth]{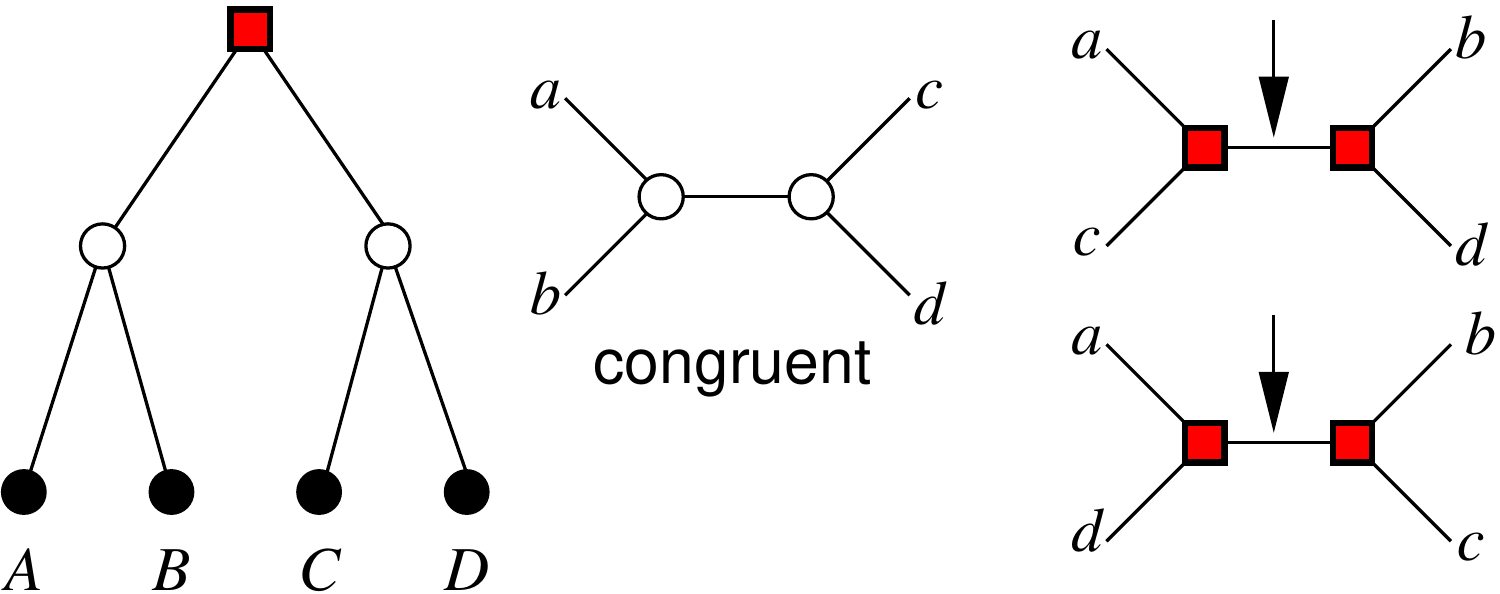}
  \end{center}
  \caption{Incongruence of gene and species quadruples implies the
    existence of an ancient duplication. Consider four pairwise distinct
    species $A$, $B$, $C$, and $D$ whose species tree is given on the
    l.h.s., and let four genes $a$, $b$, $c$, and $d$ be chosen such that
    $\sigma(a)=A$, $\sigma(b)=B$, $\sigma(c)=C$, and $\sigma(d)=D$. The two
    speciation events separating $A$ from $B$ and $C$ from $D$ are
    indicated by $\bigcirc$. The root of this tree is indicated by
    $\blacksquare$. Of the three possible gene quartets, one is congruent
    with the species tree. The other two are incongruent. In each of these,
    Equ.\ (\ref{eq:lca}) implies that the two interior vertices in these
    quartets cannot be mapped to the species tree below the root. The root
    of the gene tree must thus be mapped above the root of the species
    tree.  }
  \label{fig:ancient}
\end{figure}

\begin{theorem}
  Let $(T,\sigma)$ and $S$ be gene and species trees, respectively, and
  $a,b,c,d\in L(T)$. Moreover, let $\sigma(a)$, $\sigma(b)$, $\sigma(c)$,
  and $\sigma(d)$ be pairwise distinct species, set
  $u:= \lca_S(\sigma(a),\sigma(b),$ $\sigma(c),\sigma(d))$,
  $v_1:= \lca_S(\sigma(a),\sigma(b))$, and
  $v_2:= \lca_S(\sigma(c),\sigma(d))$. If $v_1\prec_S u$, $v_2\prec_S u$
  and $\unrooted{T}[a,b,c,d]=(ac|bd)$ or $\unrooted{T}[a,b,c,d]=(ad|bc)$,
  then $u\prec_S \mu(\lca_T(a,b,c,d))$ for every reconciliation map
  $\mu:V(T)\to V(S)\cup E(S)$ without HGT events. In particular,
  $\lca_T(a,b,c,d)$ is a duplication event.
  \label{thm:ancient}
\end{theorem}
\begin{proof}
  By assumption, $S[\sigma(a),\sigma(b),\sigma(c),\sigma(d)]$ has the
  topology shown in Fig.~\ref{fig:ancient}.  Assuming $(ac|bd)$, Equ.\
  (\ref{eq:lca}) implies
  $\mu(\lca_T(a,c))\succeq \lca_S(\sigma(a),\sigma(c))=u$ and
  $\mu(\lca_T(b,d))\succeq \lca_S(\sigma(a),\sigma(c))=u$. Thus both inner
  nodes $p$ and $q$ of the quartet are mapped no lower than $u$. The edge
  between them therefore must be mapped to an edge pre-dating $u$, since
  the speciation constraint \AX{(R3)} implies that two $\prec_T$-comparable
  events in $T$ of which one is a speciation cannot by mapped to the same
  vertex of $S$. Thus $u \prec_S \mu(\lca_T(a,b,c,d))$.  The case $(ad|bc)$
  is handled by an analogous argument exchanging $c$ and $d$. The fact that
  $\lca_T(a,b,c,d)$ is a duplication event now follows from Lemma
  \ref{lem:dupli}.
\end{proof}
This theorem can be used to discard suspicious outgroups: If
$\unrooted{T}[x,y,z_1,z_2]$ is incongruent with the known species tree,
then $\sigma(z_1)\ne\sigma(z_2)$ should be replaced by outgroup candidates
from earlier-branching species. The downside of using Theorem
\ref{thm:ancient} is that it requires the systematic investigation of a
possibly large numbers of quartets.

We suspect that it is possible in most cases to unambiguously identify
pairs whose last common ancestor in the gene tree pre-dates the last common
ancestor of the species tree under consideration. While it may be difficult
to determine the relative order of such duplications, we suspect that
clustering methods used to extract groups of co-orthologs (COGs) can be
adapted to disentangle such ancient ``paralog groups''.

\section{Simulation Results} 

Well curated gene family histories are not available at large scale. We
therefore use simulated data to evaluate how well best matches (in the
sense of evolutionary relatedness) can be estimated from both perfect
and noisy evolutionary distance measurements. For this purpose, it is
important to have data sets that emphasize asymmetric rate variations among
paralogs, i.e., the situations in which sequence dissimilarities and
divergence times are not well correlated. We therefore developed a
simulation system (see Methods) that can produce this type of test data and
explicitly records the gene family history.
  
\begin{figure}
  \begin{center}
    \includegraphics[width=0.75\columnwidth]{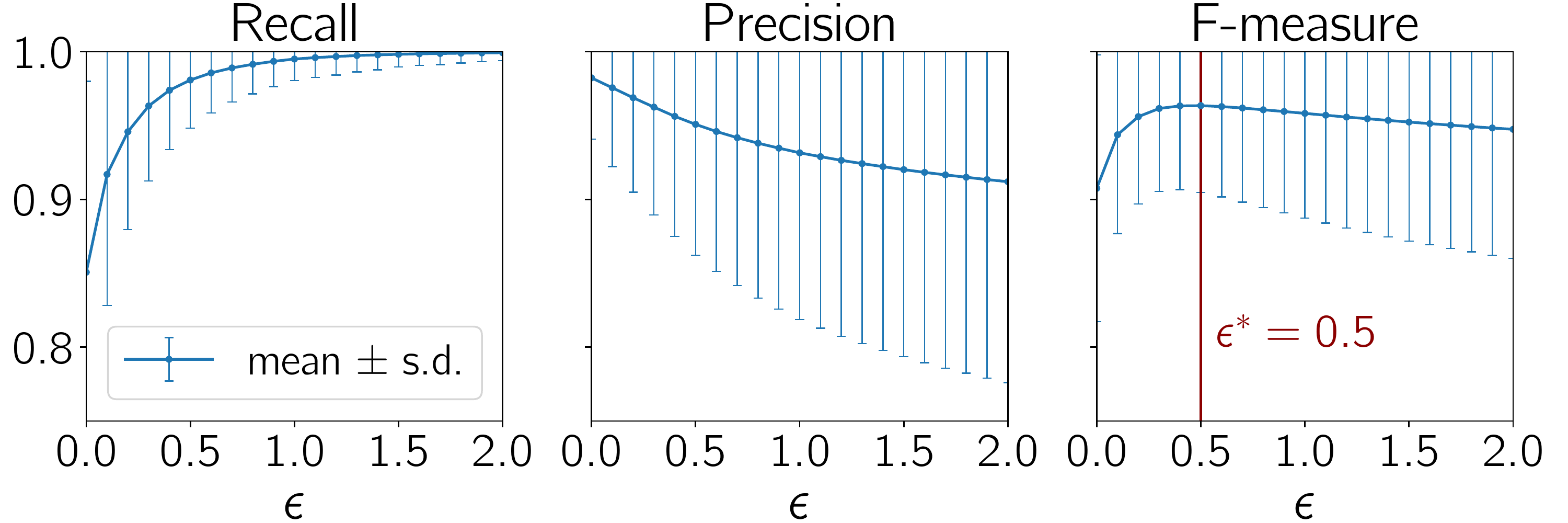}
  \end{center}
  \caption{Recall, precision, and F-measure of the true best matches as a
    function of $\epsilon$ for simulated data (2000 scenarios).}
  \label{fig:epsilon}
\end{figure}

We compare three strategies to estimate best matches:
\par\noindent\textbf{1.\ Reciprocal best hits} are inferred directly
from the distance data. In order to account for rate variations among
paralogs, we follow the strategy of \texttt{ProteinOrtho}
\cite{Lechner:11a} and consider nearly co-optimal best hits by considering
for a given gene $x$ in species $\sigma(x)$ all those $y\in Y$ as almost best
hits that have distance not worse than a factor
$1+\epsilon$ than the most similar gene in $Y$.  In symbols:\\
$H(Y|x):= \{ y\in Y\mid d(x,y)\le (1+\epsilon)\min_{y'\in Y} d(x,y') \}$\\
For further comparison we then chose the value of $\epsilon^*$ that
maximizes the F-measure ($\epsilon^*=0.5$, see Fig.\ \ref{fig:epsilon}).
Still this approach produces a substantial number of both false positives
and false negatives in data sets with substantial rate variations. We
expect that the optimal value of $\epsilon^*$ will depend on the details of
the data set, in particular on the extent of evolution rate asymmetries. In
general these will have to be estimated from the gene family history.  We
refer to this approach as the ``$\epsilon$-method''. Since we chose the
cut-off $\epsilon^*$ to maximize the $F$-measure, we effectively determine
an upper bound on the performance of the Best Hit approach.

\par\noindent\textbf{2.\ Explicit reconstruction of $\unrooted{T}$.} 
Since the additive distances completely determine $\unrooted{T}$, the only
source of errors for perfect data is an incorrect position of the root of
$T$. For additive distance data, the Neighbor Joining algorithm
\cite{Saitou:87} is guaranteed to produce the correct $\unrooted{T}$
\cite{Atteson:99}. We then use midpoint rooting \cite{Hess:07} to pass from
$\unrooted{T}$ to $T$ and compute the best matches in $T$. We refer to this
method as ``NJ+midpoint rooting''. This method is not intended as a viable
means of analysis for real-life sequence data. It serves, however, as
  a convenient way to assess the effects of rate imbalances because it
  isolates the errors that are introduced by the choice of the root alone
  i.e., by rate imbalances.

\begin{figure}
  \begin{center}
    \includegraphics[width=0.75\columnwidth]{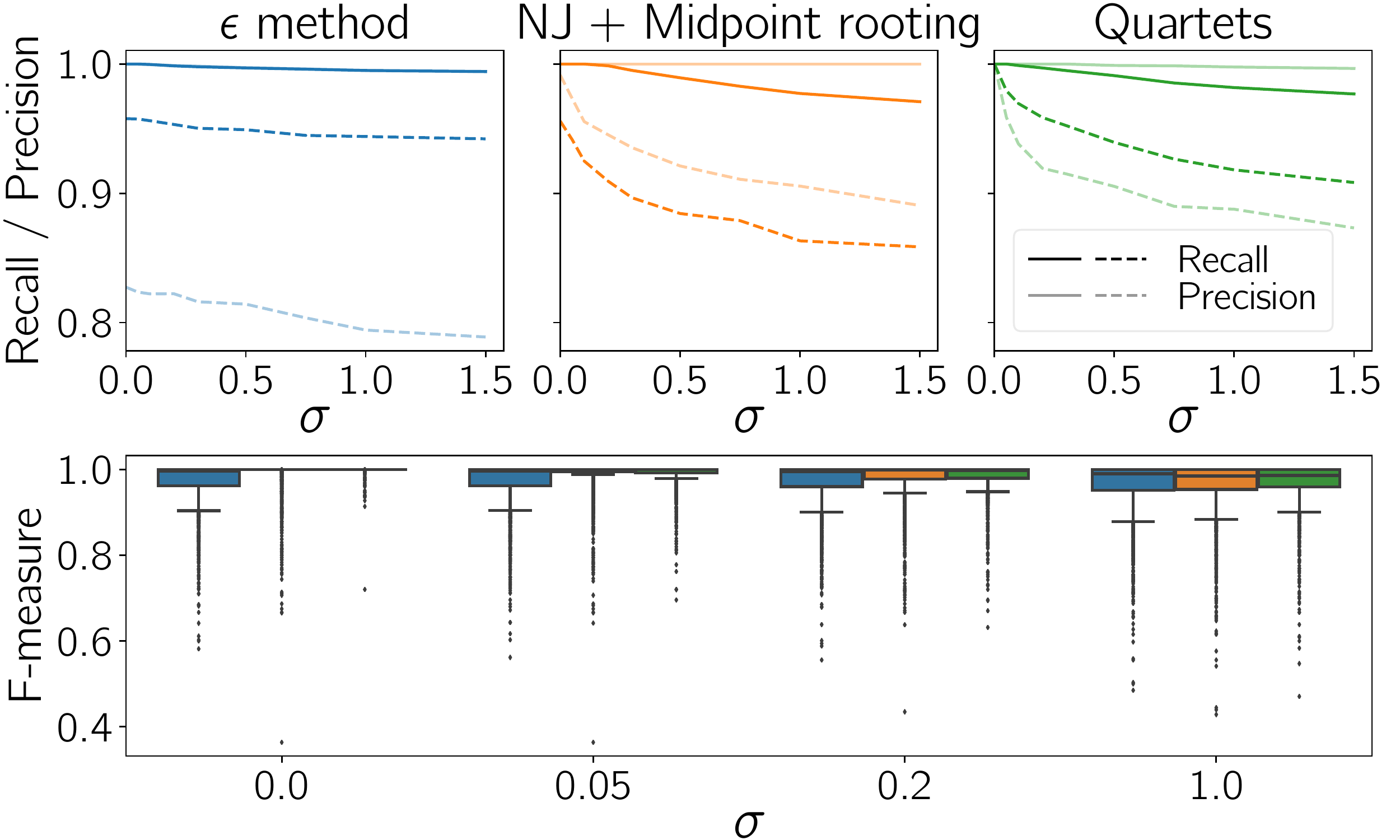}
  \end{center}
  \caption{Performance comparison of the best match inference methods for
    simulated data (2000 scenarios). Top panel: Median (solid) and
    10\textsuperscript{th} percentile (dashed) of recall and precision as a
    function of noise level. Lower panel: Boxplots of F-measure for
    different levels of noise superimposed on the additive distance;
    $\sigma=0$ refers to perfect data.  Blue: $\epsilon$ method, orange:
    explicit construction of the unrooted tree $\unrooted{T}$ and midpoint
    rooting, green: inference of quartets with outgroups chosen in another
    branch of the root.}
  \label{fig:methodcomp}
\end{figure}

\par\noindent\textbf{3.} 
The \textbf{``\emph{Quartet}'' approach} starts from a known species tree
$S$.  For $x\in X$, and $y',y''\in Y$ we pick the outgroup $z$ from a
species $\sigma(z)$ such that
$\lca_S(\sigma(X),\sigma(Y))\prec\lca_S(\sigma(X),\sigma(Y),\sigma(z))$ and
then use Algorithm \ref{alg:overall}. We explored two different ways of
determining the quartet relation: (a) After transforming the measured
sequences distances into an approximately additive evolutionary distance
(see Methods for details), we used Equ. (\ref{eq:quartet}), and (b) for
sequence data we used the direct method described in the Methods section,
thus bypassing the computation of distances altogether.

In order to benchmark the inference of best matches we compute recall and
precision w.r.t.\ the true best matches restricted to pairs of gene sets
$X$ and $Y$ for which such outgroups are available. The comparison in Fig.\
\ref{fig:methodcomp} (bottom panel) shows that the quartet method
outperforms the alternatives for different levels of the simulated
measurement error. The results are robust over a wide range of
  simulated measurement error. Not surprisingly, the reconstruction of
Neighborjoining trees already provides better results than the
$\epsilon$-method. The simple midpoint rooting strategy however still
incurs noticable level of error. For the quartet method operating on
noiseless data the only source of errors are bad choices of outgroups,
which are the consequence of ancient duplications.  The number of ancient
duplication exceeds $1$ in $5.15\%$ of the simulated gene family
scenarios. Due to loss events predating the root of the species tree, the
condition in Lemma \ref{lem:genes_in_rhoS} is only violated in $3.7\%$ of
the gene trees. Out of these problematic cases, little more than half
($2.25\%$) actually result in a non-perfect inference accuracy.

Moreover, we investigated how the number of duplication and losses
influences the inference of false positives inferred by QM (see Fig.~\ref{fig:seqmethcomp}). As expected, the number of false positives
increases with increasing duplication and loss events.

Restricting the choice of outgroup genes $z$ to species that are
  separated from $X$ and $Y$ by the root of the species tree, i.e., such
  that $\lca_S(\sigma(z),\lca_S(\sigma(X),\sigma(Y)))=\rho_S$, is likely to
  be problematic whenever $S$ is skewed in a way that leaves very few choices
  for $\sigma(z)$ and whenever the divergence between $\sigma(z)$ and
  $\lca_S(\sigma(X),\sigma(Y))$ is large. In the latter case, saturation
effects may impair the quartet inference in fast evolving gene
families. Hence, it would be advantageous to consider also genes from
closer species. In principle, every relative outgroup w.r.t.\ species
  $\sigma(X)$ and $\sigma(Y)$ is a viable candidate. These can then be
  filtered by applying Theorem \ref{thm:ancient} to reduce the number of
bad choices of $z$. We find that filtering for outgroups with identifiable
ancient duplications and giving preferences to the closest outgroup
  genes, i.e., those with the lowest
  $\lca_S(\sigma(z),\lca_S(\sigma(X),\sigma(Y)))$ indeed yields a further
moderate improvement of the estimated best matches (see Fig.\
\ref{fig:outgroup-methods}). However, the performance is slightly reduced
for perfectly additive data due to ancient duplications that where not
detected by the currently available filtering heuristics.

\begin{figure}[t]
			\begin{center}
				\includegraphics[width=0.3\textwidth]{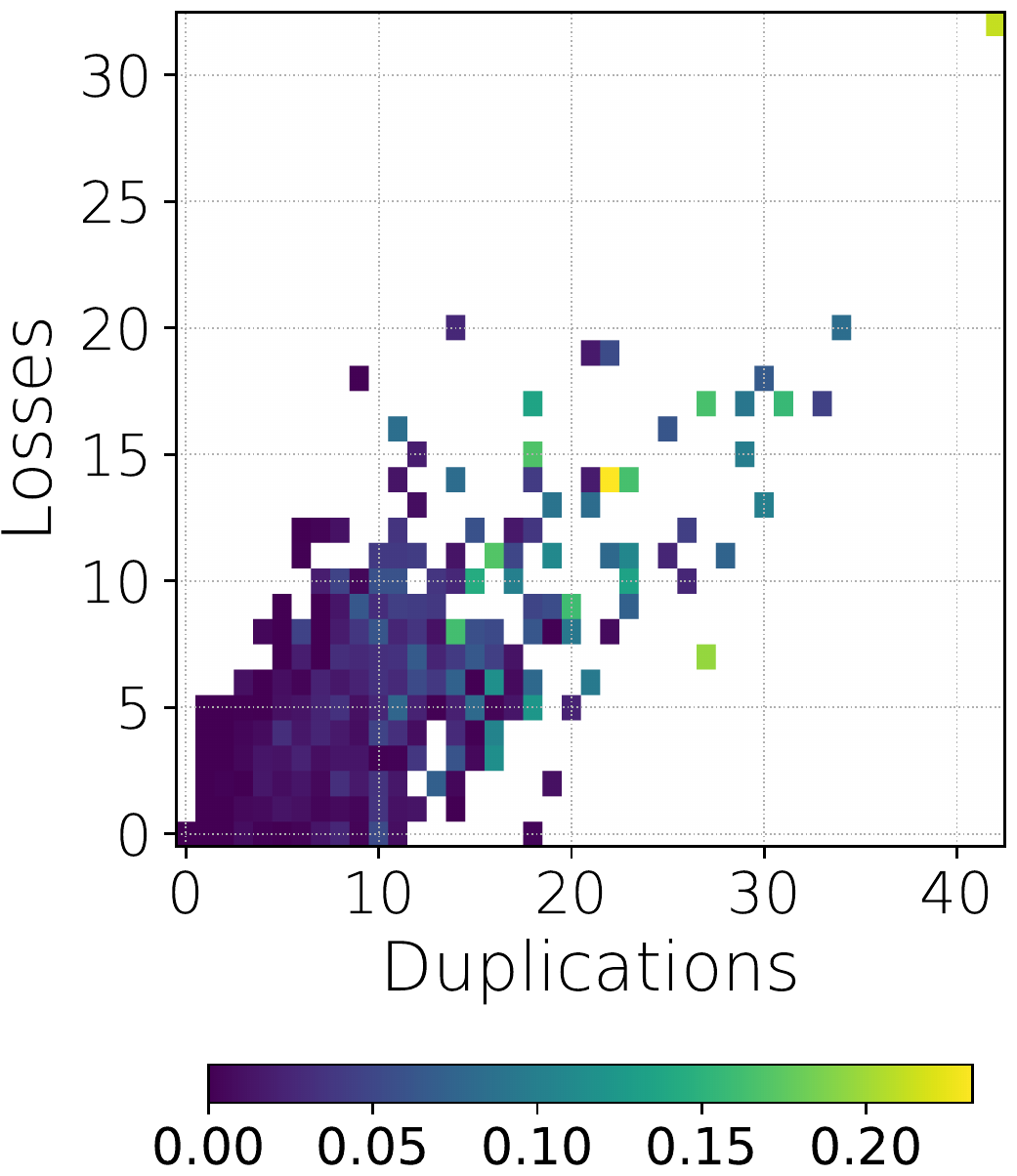}
			\end{center}
			\caption{Inference of best matches from simulated sequence data.  Heat
				map of the fraction of false positive best matches inferred by QM
				as a function of the number of duplication and loss events in the
				simulated scenario. The false positive rate is computed relative
				to the number of true best matches.}
			\label{fig:seqmethcomp}
\end{figure}

\begin{figure}[t]
  \begin{center}
    \includegraphics[width=0.75\columnwidth]{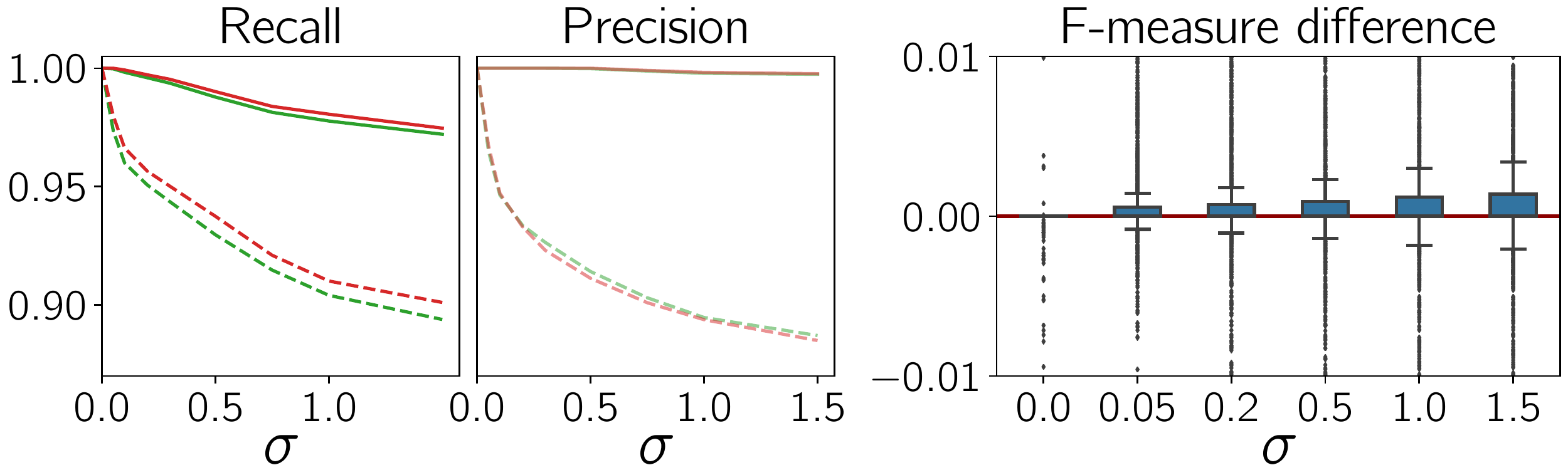}
  \end{center}
  \caption{Performance comparison of two different outgroup choice methods:
    (1) outgroups chosen randomly from species in another branch of the
    root (green, same as in Fig.\ \ref{fig:methodcomp}), (2) closest
    outgroups in all relative outgroup species corrected with Theorem
    \ref{thm:ancient} (red). Left and middle: Median (solid) and
    10\textsuperscript{th} percentile (dashed) of recall and precision as a
    function of noise level. Right: Boxplots of the F-measure differences
    for selected noise levels (method 2 $-$ method 1).}
  \label{fig:outgroup-methods}
\end{figure}

\begin{figure*}
  \begin{center}
    \includegraphics[width=\textwidth]{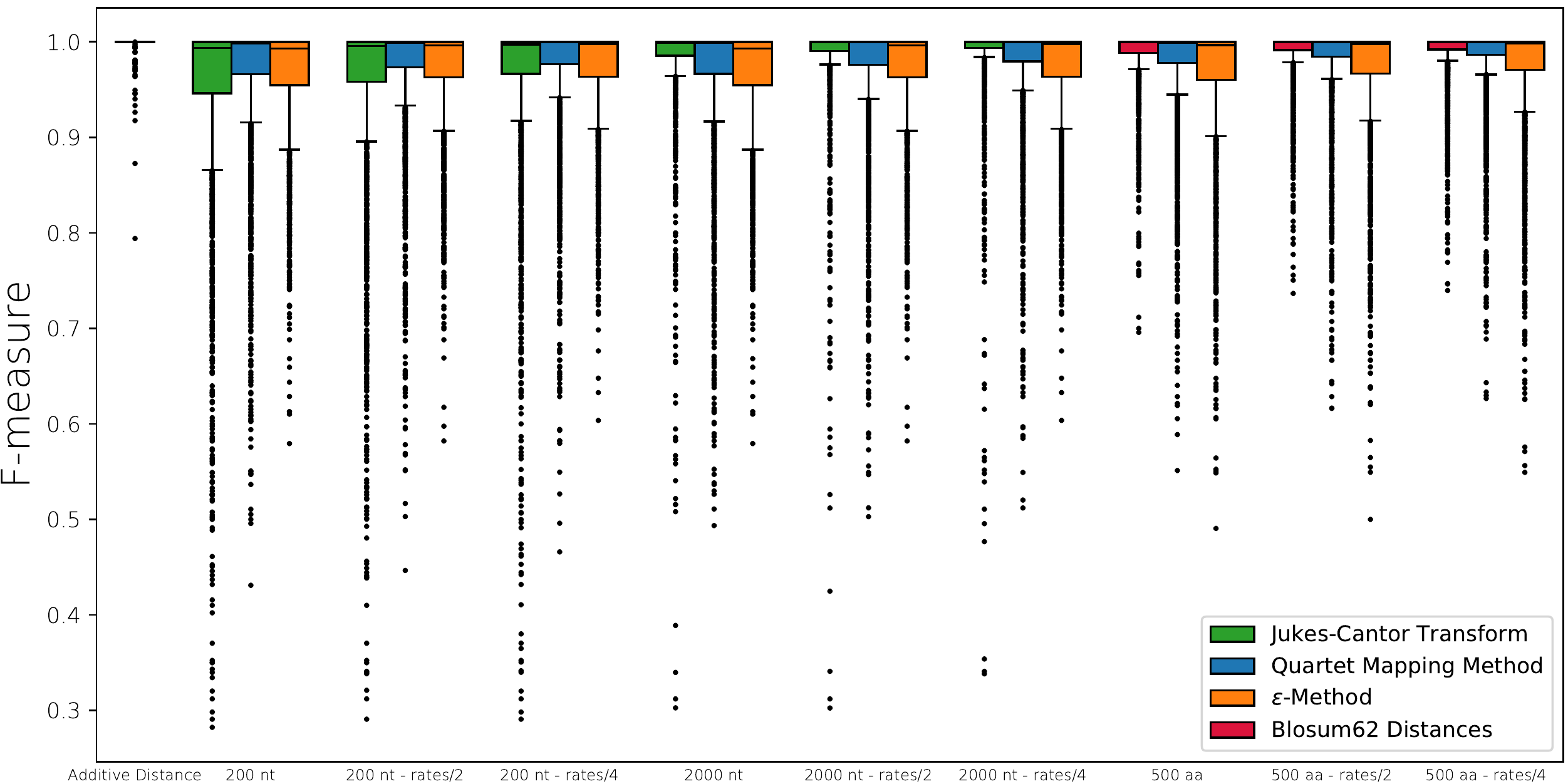}
    \end{center}
    \caption{Comparison of best match estimates for short nucleic acid
      sequences ($200$ nt), long nucleic acid sequences ($2000$ nt), and
      amino acid sequences ($500$ aa). As expected the F-measure improves
      when saturation is decreased by down-scaling the total number of
      events, and with increasing sequence length. The disappointing
      performance of the QM method is probably the consequence of a
      majority voting procedure that is too simple-minded. So far, the best
      results are obtained by estimating additive distances from pairwise
      sequence alignments using the Jukes-Cantor tranform (for nucleic acid
      sequences) or a BLOSUM-based transformation (for aminoacid
      sequences). }
    \label{fig:perf}
\end{figure*}

In applications to real-life data sets, additional uncertainties arise
through the reconstruction of distances from sequences. We therefore
simulated sequences with and without in/dels from the gene tree/species
tree scenarios and inferred the best matches from the sequence data. We
compared the results obtained (a) by transforming the Hamming distances
using the Jukes-Cantor transform and (b) by direct inference of the
quadruples using quartet mapping as outlined in the Methods section.

Fig.~\ref{fig:perf} summarized the results for simulated nucleic acid and
aminoacid sequences of different lengths and different scaling of the
evolutionary rates.  As expected, the short sequences incur a relative
large noise level compared to the perfect additive distances. Nevertheless,
the overwhelming majority of best matches is still estimated correctly
($F$-measures well above $0.9$ for the vast majority of scenarios even for
nucleic acid sequences as short as $200$ nt). Larger false positive rates
are observed only in a small number of scenarios with many duplications and
losses. This is not surprising since our relatively simple rule for
outgroup choice tends to fail if there are many ancient duplications. As
expected, the F-measure improves with increasing sequence length due to the
increased amount of information from which the distances are estimated. The
same trends were observed for simulated protein sequences

The performance of the Quartet Mapping strongly depends on how those
  quartets are handled for which none of the three possible splits
  dominates. In the default setting, these are interpreted as unresolved
  tree ($\times$) and inserted as bi-directional edges into the auxiliary
  graph $\Gamma$. This, however, leads to a moderate overprediction of best
  matches. Alternatively, a consensus can be taken over multiple choices of
  the outgroup $z$. Finally, the unresolved quadruples can be omitted
  altogether in the construction of the auxiliary graph $\Gamma$. Both
  alternatives perform worse than the default method, see Additional File
  1.

\section{Discussion and Conclusions}

The idea to use quartet structures for improvement of orthology
  estimates is not new; it was used e.g.\ in \texttt{QuartetS}
  \cite{Yu:11}. Quartets are also used as witnesses of non-orthology in OMA
  \cite{Train:17} to avoid some types of false-positives. Here, we
  systematically investigate how and when quartets help to improve and/or
correct empirical best-hit data to identify best matches in the sense of
closest evolutionary relatives. We propose that reciprocal best
  matches, rather than the uncorrected reciprocal best hits, should then be
  used to infer orthology relationships. This second step is the topic of
  an independent manuscript \cite{cobmg} in which the mathematical
  connections between (reciprocal) best matches and orthology are
  elucidated in detail.

The key observation of the present contribution is that the best
  matches of a gene $x$ in the set $Y$ of genes from a different species
  can be computed correctly if for every $y',y''\in Y$ one can find a gene
  $z$ from a third species that is an outgroup for $\{x,y',y''\}$.  From a
  theoretical point of view, this condition is closely related to rooting the
  gene tree. The second necessary ingredient is an estimate of an additive
  evolutionary distance between the genes that is accurate enough to
  correctly identify the topology of a certain subset of quadruples. We
  emphasize that this is a much less stringent condition compared to the
  ability of reconstructing the complete gene tree $T$.

Empirically, we observe that (partial) knowledge of the species tree (more
precisely: reliable monophyletic groups) is very useful for the
  choice of outgroup genes $z$: excellent results are obtained by choosing
  a candidate $z$ from a species that is an outgroup for $\sigma(x)$ and
  $\sigma(Y)$. The results can be further improved by using filtering
  criteria that identify ancient duplication events and by computing a
  consensus over serveral choices of $z$. In data sets with little
measurement noise, we indeed obtain nearly perfect best match
estimates. The theoretical considerations outlined here also suggest
additional in-roads for further improvements by means of identifying
ancient duplications, which not only serve as ``witnesses of
non-orthology'' but can also be used to prune the set of candidate
outgroups. 

In order to make the methods described here applicable to very large
  real-life data sets, it will be necessary to optimize the computational performance. To this
  end, we will develop heuristic rules to prune the set $Y$ in the case of
  large gene families. An obvious candidate is to use the $\epsilon$-method
  as an initial filter, where $\epsilon$ is now chosen to optimize the
  tradeoff between $|Y|$ and false negative predictions of best matches. We
  expect that the heuristic rules for choosing the set $Z$ of candidate
  outgroups can also be improved substantially. We expect that methods for
  orthology assessment can be improved in both reliability and
  computational performance by combining the accurate estimation of best
  matches described here with a better understanding of (reciprocal) best
  match graphs \cite{Geiss:18x,rbmg-19} and their connection with the
  orthology relation \cite{cobmg}. Since tree-free methods for orthology
  detection rely on (pairwise) best hits as proxy for reciprocal best
  matches, we expect that the accuracy of most tools would improve if best
  matches are supplied as input data. This is not easy to test, however,
  since the best hit computation is usually an integral part of the the
  software. Such a benchmark study is hence beyond the scope of this
  contribution.

The work reported here is primarily intended to provide a solid theoretical
foundation for the construction of improved best match heuristics. The
theoretical results give some guarantees for obtaining the correct best
matches and highlight some limitations that cannot be overcome with
certainty as long as only distance data are available.

\section{Methods}

\subsection{Simulations of \emph{dated} Species Trees}

As in previous work \cite{cobmg}, we use the Innovation Model
\cite{Keller:2012} to produce realistic topologies for the planted species
tree $S$. We then construct a dating function $\tau:V(S)\to[0,1]$ such that
$\tau(0_S)=1$ and $\tau(x)=0$ for $x\in L(S)$.  In order to assign a date
to an interior vertex, we traverse $S$ top-down, more precisely for the
current node $u$ at time $\tau(u)$ we proceed as follows:\\
(1) We pick a child $v\in\child(u)$ and a leaf $x\in L(S(v))$ in the
subtree below $v$. If $v$ is already a leaf, we set $\tau(v)=0$ and
proceed to the next child of $u$.\\
(2) Otherwise, we determine the number $k$ of speciations on the path
between $v$ and $x$. Hence, the path from $u$ to $x$ comprises $k+2$
edges.\\
(3) We pick a random number $r$ with mean $1$ and range $(0,2)$ from a
uniform distribution and set $\tau(v) = \tau(u) (1-r/(k+2))$ This rule is
chosen so that the expected time elapsed along the edge $uv$ equals
$\tau(u)$ divided by the number $(k+2)$ of edges along the path to
the root and ensures that $\tau(v)>0$.\\
The result is a dated species tree in which each edge $uv$ has length
$\tau(u)-\tau(v)$.

\subsection{Simulation of gene trees in the dated species tree}

We use the Gillespie algorithm \cite{Gillespie:77} to simulate the
duplication, loss and horizontal gene transfer events (HGT) occurring in
$S$. The branches of the species tree $S$ are independent in
Duplication/Loss scenarios. However, horizontal gene transfer introduces
dependencies between them. We therefore have to simulate the evolution
process in such a way that at each time point $\tau$ the possible reactions
are given by the Cartesian product $G(\tau)\times\{D,L,H\}$, where
$g\in G(\tau)$ is a gene that is present at time $\tau$ in any one of the
branches of the dates species tree, and $q\in\{D,L,H\}$ is one of the three
possible events (Duplication, Loss, HGT). Every possible simulation event
$\xi:=(g,q)$ is associated with a rate $r_{\xi}(\tau)$ that may depend
explicitly on the point in time. Rate constants are
described below.

In each step, two random numbers $r_1$ and $r_2$ are drawn independently
from the uniform distribution on $[0,1]$. The first random number $r_1$ is
used to select $\xi$ with probability $r_{\xi}(\tau)/R(\tau)$, where
$R(\tau)$ is the sum of the rates of all reactions available at time
$\tau$.  We refer to \cite{Gillespie:77} for a convenient way to implement
the rate-proportional choice of the ``reaction channel''.
Depending on the selected event type, the following actions are performed:\\
($q=L$) Gene loss is modeled by removing $g$ from the list of active
genes. \\
($q=D$) Gene duplications are modeled by placing a copy $g'$ of $g$
into the same branch of $S$ at time $\tau$.\\
($q=H$) For HGT the copy of $g'$ is placed into a different branch of
$S$. The ``landing site'' for the HGT copy is chosen uniformly from the
branches of $S$ available at time $\tau$ with the exception of the branch
harboring the parental gene $g$.

The rules determining the rate parameters for gene copies $g'$ and the
optional adjustment of rates for the genes $g$ are discussed below. The
second random variable $r_2$ is used to update the clock according to
$\tau\leftarrow \tau-\Delta\tau$ with $\Delta\tau=\ln(1/r_2)/R$. The
simulation terminates as soon as $\tau-\Delta\tau\le0$.

A complication arises from the fact that the time interval
$[\tau,\tau-\Delta\tau]$ may contain a speciation event at time
$\tau_s$. At a speciation, the gene content is copied into the
daughter-lineages, and the rates are modified in a lineage-specific manner.
As a consequence, the waiting time $\Delta\tau$ has to be re-estimated
since the set of reaction channels has changed. More precisely, we need to
determine the distribution of waiting times from a time point $t_0$ until
the next event conditioned on the fact that no event occurred between $t_0$
and $t_1$, where $t_1$ designates the time point of the speciation. For the
complementary cumulative distribution function and $s \coloneqq t_0-t_1$ we
have
\begin{align*}
  \mathbb{P}(T \ge s+t |  T \ge s) & =
  \mathbb{P}(T \ge s+t \wedge T \ge s)/\mathbb{P}(T\ge s) \\
   & = \mathbb{P}(T \ge s + t)/\mathbb{P}(T \ge s)
\end{align*}
Since the waiting time distributions are exponential with rate $r_1$
before $t_1$ and rate $r_2$ following the speciation event, we obtain
\begin{equation*}
  \mathbb{P}(T \ge s+t |  T \ge s) = e^{-(r_1 s + r_2 t)}/e^{-r_1 s} =
  e^{-r_2t}  
\end{equation*}
Hence, if the simulated waiting time reaches beyond the speciation event,
the clock is advanced to the speciation event and a new waiting time is
drawn with the rates after the speciation event. In practice, a new random
number to obtain the time step $\Delta\tau'$ with the updated rates after
the speciation event. In the new interval $[\tau_S,\tau_S-\Delta\tau']$ we
again have to check for speciation events. Since the speciation events are
known \emph{a priori} from the dated species tree $S$, they are held in a
priority queue in temporal order. The final result is a dated gene tree
$T$, i.e., each event is unambiguously associated with a time stamp. The
simulation also completely determines the reconciliation map $\mu$.

We simulated 2000 pairs of species and gene trees, where $|L(S)|$ was drawn
uniformly from the interval $[3,50]$. The duplication and loss rates were
(independently) drawn from $[0.5,1.0)$.

\subsection{Modeling rate imbalances}

In order to produce realistic (sequence) data, an evolution rate $\omega_e$
has to be associated with each edge $e$ of $T$. To this end we use a
hierarchical model that first determines a baseline gene substitution rate
$\omega^0_e$ for each edge $e$ of the species tree $S$ in order to simulate
effects such as variations of population size and generation time. This
introduces a correlation between the rates of all genes in the same lineage
of $S$. These base rates are then modified by gene-specific contributions
that capture effects such as differences in selection pressures that depend
on gene function and rate differences in the wake of duplications such as
neofunctionalization and subfunctionalization \cite{Force:99}. In detail,
we use the following parametrization:\\
$\bullet$  mean substitution rate of the conserved members of a gene family
  (default $1.0$).\\
$\bullet$ variance $\sigma_0^2$ for the baseline substitution rate in $S$
  (default $0.2$).\\
$\bullet$ a gamma distribution for the substitution rates $>1$ of divergent
  genes.  The parameters are estimated from data for the whole genome
  duplication in saccharomycete yeasts \cite{Byrne:07}. Alternatively,
  a uniform distribution on $(1; r_{max}]$ can be selected.\\
$\bullet$ weights for the relative frequency of the possible fates of
  duplicates (functional conservation, subfunctionalization,
  neofunctionalization; default equal weights $1/3$).

We determine the baseline substitution rates $\omega_{uv}^0$ for the edge
$uv \in E(S)$ as follows: We simply assign the mean substitution rate to
the planted edge $0_S\rho_S$ (i.e. 1.0 by default).  We traverse $S$ in
pre-order and draw for each edge $uv \in E(S) \setminus \{0_S\rho_S\}$ the
logarithm $\ln \omega_{uv}^0$ of the rate of evolution from a normal
distribution with variance $\sigma^2 = \sigma_0^2(\tau(u)-\tau(v))$.  To
avoid bias towards higher or lower rates, we normalize the mean of the
normal distribution such that $E(\omega_{uv}^0 )= \omega_{\parent(u)u}^0$.

For the gene specific rates we first sort all vertices $u \in V(T)$ by
$\tau(u)$ in descending temporal order. We keep track of the current number
of genes in each branch of the species tree. During the simulation, the
edges of $T$ will be marked as either \texttt{conserved} or
\texttt{divergent} depending on the fate of the branch after a duplication
event. For each edge $e=uv \in E(T)$ in the gene tree, we initialize an
empty list $\mathfrak{L}_e$ of ordered pairs of the form $(\tau,\omega)$ to
record the gene-specific evolution rates $\omega$ and the corresponding
time points $\tau$ at which they become valid during the existence of
$e$. This allows us to reset the \texttt{divergent} status of a gene in
case it is the last survivor in a given species. At present, we do not
consider other events that change the rate of evolution of a gene within
the edge $e$. The framework, however, can easily accommodate such rules in
future refinements of the model. We denote by $\mathfrak{L}_{e,i}$ the
i\textsuperscript{th} ordered pair $(\tau_i,\omega_i)$ in $\mathfrak{L}_e$
and define $\tau(\mathfrak{L}_{e,i}) \coloneqq \tau_i$ and
$\omega(\mathfrak{L}_{e,i}) \coloneqq \omega_i$.

Recall that $0_T\rho_T$ is the first (planted) edge in $T$. To initialize
the simulation, we mark $0_T\rho_T$ as \texttt{conserved} and append
$(\tau(0_T), 1.0)$ to $\mathfrak{L}_{0_T\rho_T}$.
Then for each vertex $u$ in the sorted list we proceed as follows:\\
\textit{(1) $u$ is a speciation event} \\
Mark all edges $uv$ with $v \in \child(u)$ the same as $\parent(u)u$. To
$\mathfrak{L}_{uv}$ we append the pair $(\tau(u), \omega)$ with
$\omega=1.0$ ($uv$ is \texttt{conserved}) or $\omega$ Gamma-distributed
($uv$
is \texttt{divergent}), respectively.\\
\textit{(2) $u$ is a duplication event} \\
If the edge $\parent(u)u$ is marked as \texttt{divergent}, then all edges
$uv$ with $v \in \child(u)$ are also marked as \texttt{divergent} and
corresponding pairs $(\tau(u),\omega)$ are appended to $\mathfrak{L}_{uv}$,
where the values of $\omega$ are drawn i.i.d.\ from
the Gamma distribution. \\
If $\parent(u)u$ is marked as \texttt{conserved}, we choose between (a)
conservation, (b) subfunctionalization and (c) neofunctionalization with
the specified weights.  For (a) mark both incident edges below $u$ as
\texttt{conserved}, for (b) as \texttt{divergent} and for (c) one edge is
\texttt{conserved} and the other is \texttt{divergent}.  To
$\mathfrak{L}_{uv}$ we append the pair $(\tau(u), \omega)$ with
$\omega=1.0$ ($uv$ is \texttt{conserved}) or $\omega$ Gamma-distributed
($uv$ is \texttt{divergent}), respectively.\\
\textit{(3) $u$ is a loss event}\\
If a single copy is left in the respective species after the loss: Let
$e^*$ be the corresponding edge of the remaining copy at $\tau(u)$. Mark
$e^*$ as \texttt{conserved} and append the pair $(\tau(u),1.0)$ to
$\mathfrak{L}_{e^*}$.\\
\textit{(4) $u$ is an HGT event}\\
Let $v_1$ be the copy that remains in the species and $v_2$ the transferred
copy. Mark $uv_1$ the same as $\parent(u)u$ and append $(\tau(u),\omega)$
to $\mathfrak{L}_{uv_1}$ where $\omega$ is the last rate that was appended
to $\mathfrak{L}_{\parent(u)u}$. Mark $uv_2$ as \texttt{divergent} and
append $(\tau(u),\omega)$ to $\mathfrak{L}_{uv_2}$ with $\omega$
Gamma-distributed.

For each edge $e=uv$ in $T$ we finalize $\mathfrak{L}_{e}$ by appending
$(\tau(v),\omega)$ where $\omega$ is the last rate that was appended to
$\mathfrak{L}_{e}$ so far. We then define the edge length $\ell(e)$
for each edge $e$ in $T$ as
\begin{equation}
  \ell(e) = \omega_{f}^0
  \sum_{i=1}^{|\mathfrak{L}_{e}|-1}\omega(\mathfrak{L}_{e,i}) \cdot
  (\tau(\mathfrak{L}_{e,i}) - \tau(\mathfrak{L}_{e,i+1}))
\end{equation}
where $f$ is the edge in the species tree $S$ into which $e$ is embedded.
The resulting distance function $\ell: E(T)\to \mathbb{R}^+$ defines an
additive metric on the set of vertices $V(T)$.  We denote by $\mathbf{D}$ a
distance matrix on the set of non-loss leaves in $T$, i.e., the extant
genes at time $\tau=0$.  

\subsection{Simulation of measurement noise}

In order to simulate measurement noise we consider three strategies:
\par\noindent(1) It is easy to see that a convex combination\\
$(1-\epsilon)\mathbf{D} + \epsilon \mathbf{D'}$, $0\le \epsilon\le 1$ of
two metrics $\mathbf{D}$ and $\mathbf{D'}$ is again a metric (i.e., in
particular, satisfies the triangle inequality). Even if both $\mathbf{D}$
and $\mathbf{D'}$ are additive, however, their convex combination is not
additive in general. This yields a distance that is affected by a
systematic bias corresponding to the noise contribution
$\epsilon\mathbf{D'}$.
\par\noindent(2) Adding i.i.d.\ random noise $\varepsilon_{xy}$ with mean
$1$ and standard deviation $s$ to the additive distance, i.e., substituting
$D'(x,y)=D'(y,x)= \varepsilon_{xy}D(x,y)$, in general violates the triangle
inequality. We therefore use the following simple algorithm: choose
$x\ne y\in L$ at random. If $D'$ satisfies the triangle inequality, then
accept the addition of $\varepsilon_{xy}$; otherwise, reject the change and
choose a new random perturbation. We repeat this until $\binom{|L|}{2}$
perturbations have been accepted. An alternative approach is to first
introduce perturbations to all distances in $\mathbf{D}$ and than to
extract a corrected distance matrix $\mathbf{\hat D}$ using one of several
algorithms for the ``metric repair problem'', see e.g.\
\cite{Brickell:08,Gilbert:17}. A cursory test showed that the trees
reconstructed from distance matrices processed with these methods tend to
be more different from the reference than with our approach of enforcing
the triangle inequality immediately. We therefore did not pursue them
further in the this contribution.
\par\noindent(3) The distance $\ell(e)$ can also be interpreted as the
number of evolutionary events per site. These are then simulated directly
using the generation tool \texttt{pyvolve} \cite{Spielman:15}. We generated
nucleic acid sequences of length $200$ with equal transition and
transversion rates and of length $2000$ with a transition:transversion
ratio of $2:1$.  To assess saturation effects, we scaled the rates $\mu_e$
by $1/4$, $1/2$, $1$, and $2$, respectively. To better connect this work
with protein-based orthology assessment pipelines, we simulated aminoacid
sequences of length $500$ using the WAG model \cite{Whelan:01}. Distances
were then estimated in \texttt{biopython} with the \texttt{BLOSUM62} matrix
\cite{Eddy:04}.

\subsection{Estimation of quartets from sequence data}

In order to estimate quartets directly from aligned sequence data, we
  use the approach of statistical geometry \cite{Eigen:88,Nieselt:97}. We
  start from a multiple alignment of the sequences $x$, $y'$, $y''$, and
  $z$, which we assume to appear in this order. The alignment itself is
  produced by the sequence simulator and thus does not need to be
  recomputed. Each alignment column belongs to one of the 15 categories
determined by which of the four sequence $x$, $y'$, $y''$, and $z$ feature
the same character:
\begin{equation*}
  \setlength{\tabcolsep}{1pt}
  \begin{tabular}{|C|CCCCCCCCCCCCCCC|}
  \hline
  &\C{1} & \C{2} & \C{3} & \C{4} & \C{5} & \C{6} & \C{7} & \C{8} & \C{9} &
   \C{10} & \C{11} & \C{12} & \C{13} & \C{14} & \C{15} \\
  \hline
  x   & a & a & a & a & b & a & a & a & a & a & a & b & b & b & a \\   
  y'  & a & a & a & b & a & a & b & b & a & b & b & a & a & c & b \\
  y'' & a & a & b & a & a & b & a & b & b & a & c & a & c & a & c \\
  z   & a & b & a & a & a & b & b & a & c & c & a & c & a & a & d \\
  \hline
  \end{tabular}
\end{equation*}
The categories $\C{1}$ through $\C{5}$ and $\C{15}$ do not convey
phylogenetic information. Of the remaining ones, $\C{6}$, $\C{9}$, and
$\C{14}$ support $(xy'|y''z)$, $\C{7}$, $\C{10}$, and $\C{13}$ support
$(xy''|y'z)$, and $\C{8}$, $\C{11}$, and $\C{12}$ support $(xz|y'y'')$
\cite{Nieselt:01}.  Denoting by $d_{aaaa}$, etc., the number of alignment
columns belonging to a given category, the support scores for
\emph{geometry mapping} \cite{Nieselt:01} are
\begin{equation}
  \begin{split}
    S(xy'|y''z) &= d_{aabb} + \frac{1}{2}(d_{aabc}+d_{bcaa}) \\
    S(xy''|y'z) &= d_{abab} + \frac{1}{2}(d_{abac}+d_{baca}) \\
    S(xz|y'y'') &= d_{abba} + \frac{1}{2}(d_{abca}+d_{baac}) \\
  \end{split}
\end{equation}
Using $S:=S(xy'|y''z)+S(xy''|y'z)+S(xz|y'y'')$, normalized scores 
are defined a $s(xy'|y''z):=S(xy'|y''z)/S$. 
This unweighted version can be extended to a weighted version when a
non-trivial distance measure $D$ on the underlying alphabet is given.  As
derived in \cite{Nieselt:01}, a support value for the three possible
quartets can be computed separately for each
alignment column $i$ as the isolation index for the distances on the four
characters:
\begin{equation}
  \begin{split}
    2\beta_i(xy'|y''z) &= D^*_i - (D(x_i,y_i') + D(y_i'',z_i)) \\
    2\beta_i(xy''|y'z) &= D^*_i - (D(x_i,y_i'') + D(y_i',z_i)) \\
    2\beta_i(xz|y'y'') &= D^*_i - (D(x_i,z_i) + D(y_i',y_i'')) \\
  \end{split}
\end{equation}
Here $D_i^*$ is the largest of the three distance sums appearing in the
equation above. Summing up the $\beta_i(\,.\,)$ values over all alignment
columns $i$ yields aggregated support scores $\beta(\,.\,)$. These are
conveniently normalized to relative values as in the unweighted case. The
relative support scores for the weighted model reduce to the unweighted
ones if $D(a,b)=1-\delta-{a,b}$ is the trivial metric \cite{Nieselt:01}.
If no quartet can be inferred unambiguously, then we default to the
assumption $\lca(x,y')=\lca(x,y'')$.

\section*{Competing interests}
  The authors declare that they have no competing interests.
  
\section*{Author's contributions}
PFS designed the study, PFS, MG, MH, DS, and MHR derived the mathematical
results, DS implemented the generation of weighted scenarios, AL, MGL, and
DIV performed the sequence-based simulations. All authors contributed to the
interpretation of results and the writing of the manuscript.

\section*{Funding}
This work was support in part by the German Federal Ministry of Education
and Research (BMBF, project no.\ 031A538A, de.NBI-RBC) and the Mexican
Consejo Nacional de Ciencia y Tecnolog{\'i}a (CONACyT, 278966 FONCICYT 2).
Publication costs are covered by the DFG through the Open Access Fund at
Universtit{\"a}t Leipzig.

\begin{appendix}    
\section*{Additional Files}
\subsection*{Additional file 1 --- Alternative construction of $\Gamma$}
In order to further investigate the inaccuracies introduced by unresolved
quartets, we considered alternative constructions of the auxiliary graph
$\Gamma$. In addition to the default method, we omitted all edges defined
for quartets classified as unresolved ($\times$), and we ignored the
contribution of outgroups that lead to unresolved and used a majority vote
only for the remaining choices of the outgroup. All non-trivial sinks
  in $\Gamma$ were then interpreted as best matches, i.e., isolated
  vertices in $\Gamma$ were ignored.  Both variants perform worse than the
default method. Data are compared for short nucleic acid sequences
  with rates of sequence divergence scaled by $1$ (OR), $1/2$ (D2), $1/4$
  (D4), and $2$ (M2).

\texttt{AdditionalFile\_1.pdf}

\end{appendix}

\bibliographystyle{plainnat}
\bibliography{amb-infer}

\begin{thebibliography}{67}
\providecommand{\natexlab}[1]{#1}
\providecommand{\url}[1]{\texttt{#1}}
\expandafter\ifx\csname urlstyle\endcsname\relax
  \providecommand{\doi}[1]{doi: #1}\else
  \providecommand{\doi}{doi: \begingroup \urlstyle{rm}\Url}\fi

\bibitem[Aho et~al.(1981)Aho, Sagiv, Szymanski, and Ullman]{Aho:81}
A.V. Aho, Y.~Sagiv, T.G. Szymanski, and J.D. Ullman.
\newblock Inferring a tree from lowest common ancestors with an application to
  the optimization of relational expressions.
\newblock \emph{SIAM J Comput}, 10:\penalty0 405--421, 1981.
\newblock \doi{10.1137/0210030}.

\bibitem[Altenhoff and Dessimoz(2009)]{Altenhoff:09}
A.~M. Altenhoff and C.~Dessimoz.
\newblock Phylogenetic and functional assessment of orthologs inference
  projects and methods.
\newblock \emph{PLoS Comput Biol}, 5:\penalty0 e1000262, 2009.
\newblock \doi{10.1371/journal.pcbi.1000262}.

\bibitem[Altenhoff et~al.(2016)Altenhoff, Boeckmann, Capella-Gutierrez,
  Dalquen, DeLuca, Forslund, Jaime, Linard, Pereira, Pryszcz, Schreiber,
  da~Silva, Szklarczyk, Train, Bork, Lecompte, von Mering, Xenarios,
  Sj{\"o}lander, Jensen, Martin, Muffato, Gabald{\'o}n, Lewis, Thomas,
  Sonnhammer, and Dessimoz]{Altenhoff:16}
Adrian~M. Altenhoff, Brigitte Boeckmann, Salvador Capella-Gutierrez, Daniel~A.
  Dalquen, Todd DeLuca, Kristoffer Forslund, Huerta-Cepas Jaime, Benjamin
  Linard, C{\'e}cile Pereira, Leszek~P. Pryszcz, Fabian Schreiber, Alan~Sousa
  da~Silva, Damian Szklarczyk, Cl{\'e}ment-Marie Train, Peer Bork, Odile
  Lecompte, Christian von Mering, Ioannis Xenarios, Kimmen Sj{\"o}lander,
  Lars~Juhl Jensen, Maria~J. Martin, Matthieu Muffato, Toni Gabald{\'o}n,
  Suzanna~E. Lewis, Paul~D. Thomas, Erik Sonnhammer, and Christophe Dessimoz.
\newblock Standardized benchmarking in the quest for orthologs.
\newblock \emph{Nature Methods}, 13:\penalty0 425--430, 2016.
\newblock \doi{10.1038/nmeth.3830}.

\bibitem[Atteson(1999)]{Atteson:99}
Kevin Atteson.
\newblock The performance of {Neighbor-Joining} methods of phylogenetic
  reconstruction.
\newblock \emph{Algorithmica}, 25:\penalty0 251--278, 1999.
\newblock \doi{10.1007/PL00008277}.

\bibitem[B{\"o}cker and Dress(1998)]{Boecker:98}
Sebastian B{\"o}cker and Andreas W.~M. Dress.
\newblock Recovering symbolically dated, rooted trees from symbolic
  ultrametrics.
\newblock \emph{Adv. Math.}, 138:\penalty0 105--125, 1998.
\newblock \doi{10.1006/aima.1998.1743}.

\bibitem[Bork et~al.(1998)Bork, Dandekar, Diaz-Lazcoz, Eisenhaber, Huynen, and
  Yuan]{Bork:98}
P~Bork, T~Dandekar, Y~Diaz-Lazcoz, F~Eisenhaber, M~Huynen, and Y~Yuan.
\newblock Predicting function: from genes to genomes and back.
\newblock \emph{J Mol Biol}, 283:\penalty0 707--725, 1998.
\newblock \doi{10.1006/jmbi.1998.2144}.

\bibitem[Brickell et~al.(2008)Brickell, Dhillon, Sra, and Tropp]{Brickell:08}
Justin Brickell, Inderjit~S. Dhillon, S.~Sra, and Joel~A. Tropp.
\newblock The metric nearness problem.
\newblock \emph{SIAM J. Matrix Analysis Appl.}, 30:\penalty0 375--396, 2008.
\newblock \doi{10.1137/060653391}.

\bibitem[Buneman(1974)]{Buneman:74}
Peter Buneman.
\newblock Note on the metric properties of trees.
\newblock \emph{J. Comb. Th. B}, 17:\penalty0 48--50, 1974.
\newblock \doi{10.1016/0095-8956(74)90047-1}.

\bibitem[Byrne and Wolfe(2007)]{Byrne:07}
Kevin~P. Byrne and Kenneth~H. Wolfe.
\newblock Consistent patterns of rate asymmetry and gene loss indicate
  widespread neofunctionalization of yeast genes after whole-genome
  duplication.
\newblock \emph{Genetics}, 175:\penalty0 1341--1350, 2007.
\newblock \doi{10.1534/genetics.106.066951}.

\bibitem[Cherlin et~al.(2018)Cherlin, Nye, Boys, Heaps, Williams, and
  Embley]{Cherlin:18}
S~Cherlin, T~M~W Nye, R~J Boys, S~E Heaps, T~A Williams, and T~M Embley.
\newblock The effect of non-reversibility on inferring rooted phylogenies.
\newblock \emph{Mol Biol Evol}, 35:\penalty0 984--1002, 2018.
\newblock \doi{10.1093/molbev/msx294}.

\bibitem[Doyon et~al.(2011)Doyon, Ranwez, Daubin, and Berry]{Doyon:11}
J-P Doyon, V~Ranwez, V~Daubin, and V~Berry.
\newblock Models, algorithms and programs for phylogeny reconciliation.
\newblock \emph{Brief Bioinform.}, 12:\penalty0 392--400, 2011.
\newblock \doi{10.1093/bib/bbr045}.

\bibitem[Drummond et~al.(2006)Drummond, Ho, Phillips, and Rambaut]{Drummond:06}
A~J Drummond, S~Y~W Ho, M~J Phillips, and A~Rambaut.
\newblock Relaxed phylogenetics and dating with confidence.
\newblock \emph{PLoS Biol}, 4:\penalty0 699--710, 2006.
\newblock \doi{10.1371/journal.pbio.0040088}.

\bibitem[Eddy(2004)]{Eddy:04}
Sean~R Eddy.
\newblock Where did the {BLOSUM62} alignment score matrix come from?
\newblock \emph{Nature Biotech.}, 22:\penalty0 1035--1036, 2004.
\newblock \doi{10.1038/nbt0804-1035}.

\bibitem[Eigen et~al.(1988)Eigen, Winkler-Oswatitsch, and Dress]{Eigen:88}
Manfred Eigen, Ruth Winkler-Oswatitsch, and Andreas W.~M. Dress.
\newblock Statistical geometry in sequence space: a method of quantitative
  comparative sequence analysis.
\newblock \emph{Proc Natl Acad Sci USA}, 85:\penalty0 5913--5917, 1988.
\newblock \doi{10.1073/pnas.85.16.5913}.

\bibitem[Fitch(1970)]{Fitch:70}
W~M Fitch.
\newblock Distinguishing homologous from analogous proteins.
\newblock \emph{Syst Zool}, 19:\penalty0 99--113, 1970.
\newblock \doi{10.2307/2412448}.

\bibitem[Fitch(1981)]{Fitch:81}
W.~M. Fitch.
\newblock A non-sequential method for constructing trees and hierarchical
  classifications.
\newblock \emph{J. Mol. Evol.}, 18:\penalty0 30--37, 1981.
\newblock \doi{10.1007/BF01733209}.

\bibitem[Fitch(2000)]{Fitch:00}
Walter~M. Fitch.
\newblock Homology: a personal view on some of the problems.
\newblock \emph{Trends Genet.}, 16:\penalty0 227--231, 2000.
\newblock \doi{10.1016/S0168-9525(00)02005-9}.

\bibitem[Force et~al.(1999)Force, Lynch, Pickett, Amores, Yan, and
  Postlethwait]{Force:99}
Allan Force, Michael Lynch, F.~Bryan Pickett, Angel Amores, Yi-lin Yan, and
  John Postlethwait.
\newblock Preservation of duplicate genes by complementary, degenerative
  mutations.
\newblock \emph{Genetics}, 151:\penalty0 1531--1545, 1999.

\bibitem[Gabald{\'o}n and Koonin(2013)]{Gabaldon:13}
Toni Gabald{\'o}n and Eugene~V. Koonin.
\newblock Functional and evolutionary implications of gene orthology.
\newblock \emph{Nat Rev Genet.}, 14:\penalty0 360--366, 2013.
\newblock \doi{10.1038/nrg3456}.

\bibitem[Gei{\ss} et~al.(2019{\natexlab{a}})Gei{\ss}, Stadler, and
  Hellmuth]{rbmg-19}
M.~Gei{\ss}, P.~F. Stadler, and M~Hellmuth.
\newblock Reciprocal best match graphs.
\newblock Technical Report arXiv:1903.07920v4, arXiv, 2019{\natexlab{a}}.

\bibitem[Gei{\ss} et~al.(2019{\natexlab{b}})Gei{\ss}, Ch{\'a}vez, Gonz{\'a}lez,
  L{\'o}pez, Stadler, Valdivia, Hellmuth, Hern{\'a}ndez~Rosales, and
  Stadler]{Geiss:18x}
Manuela Gei{\ss}, Edgar Ch{\'a}vez, Marcos Gonz{\'a}lez, Alitzel L{\'o}pez,
  B{\"a}rbel M~R Stadler, Dulce Valdivia, Marc Hellmuth, Maribel
  Hern{\'a}ndez~Rosales, and Peter~F Stadler.
\newblock Best match graphs.
\newblock \emph{J. Math. Biol.}, 78:\penalty0 2015--2057, 2019{\natexlab{b}}.
\newblock \doi{10.1007/s00285-019-01332-9}.

\bibitem[Gei{\ss} et~al.(2019{\natexlab{c}})Gei{\ss}, Gonz{\'a}lez~Laffitte,
  L{\'o}pez~S{\'a}nchez, Valdivia, Hellmuth, Hern{\'a}ndez~Rosales, and
  Stadler]{cobmg}
Manuela Gei{\ss}, Marcos~E.\ Gonz{\'a}lez~Laffitte, Alitzel
  L{\'o}pez~S{\'a}nchez, Dulce~I. Valdivia, Marc Hellmuth, Maribel
  Hern{\'a}ndez~Rosales, and Peter~F. Stadler.
\newblock Best match graphs and reconciliation of gene trees with species
  trees.
\newblock Technical Report arXiv:1904.12021v2, arXiv, 2019{\natexlab{c}}.

\bibitem[Gilbert and Jain(2017)]{Gilbert:17}
Anna~C. Gilbert and Lalit Jain.
\newblock If it ain't broke, don't fix it: Sparse metric repair.
\newblock In \emph{55th Annual Allerton Conference on Communication, Control,
  and Computing}, pages 612--619, 2017.
\newblock \doi{10.1109/ALLERTON.2017.8262793}.

\bibitem[Gillespie(1977)]{Gillespie:77}
Daniel~T. Gillespie.
\newblock Exact stochastic simulation of coupled chemical reactions.
\newblock \emph{J. Phys. Chem.}, 81:\penalty0 2340--2361, 1977.
\newblock \doi{10.1021/j100540a008}.

\bibitem[G{\'o}recki and Tiuryn(2006)]{Gorecki:06}
Pawe{\l} G{\'o}recki and Jerzy Tiuryn.
\newblock {DLS}-trees: A model of evolutionary scenarios.
\newblock \emph{Theor. Comp. Sci.}, 359:\penalty0 378--399, 2006.
\newblock \doi{10.1016/j.tcs.2006.05.019}.

\bibitem[Hasegawa et~al.(1985)Hasegawa, Kishino, and Yano]{Hasegawa:85}
M~Hasegawa, H~Kishino, and T~Yano.
\newblock Dating of human-ape splitting by a molecular clock of mitochondrial
  {DNA}.
\newblock \emph{J. Mol. Evol.}, 22:\penalty0 160--174, 1985.
\newblock \doi{10.1007/BF02101694}.

\bibitem[Hellmuth(2017{\natexlab{a}})]{Hel-17}
M.~Hellmuth.
\newblock Biologically feasible gene trees, reconciliation maps and informative
  triples.
\newblock \emph{Alg. Mol. Biol.}, 12:\penalty0 23, 2017{\natexlab{a}}.
\newblock \doi{10.1186/s13015-017-0114-z}.

\bibitem[Hellmuth(2017{\natexlab{b}})]{Hellmuth2017}
Marc Hellmuth.
\newblock Biologically feasible gene trees, reconciliation maps and informative
  triples.
\newblock \emph{Alg. Mol. Biol.}, 12:\penalty0 23, 2017{\natexlab{b}}.
\newblock \doi{10.1186/s13015-017-0114-z}.

\bibitem[Hernandez-Rosales et~al.(2012)Hernandez-Rosales, Hellmuth, Wieseke,
  Huber, Moulton, and Stadler]{HernandezRosales:12a}
Maribel Hernandez-Rosales, Marc Hellmuth, Nicolas Wieseke, Katharina~T. Huber,
  Vincent Moulton, and Peter~F. Stadler.
\newblock From event-labeled gene trees to species trees.
\newblock \emph{BMC Bioinformatics}, 13\penalty0 (Suppl.\ 19):\penalty0 S6,
  2012.
\newblock \doi{10.1186/1471-2105-13-S19-S6}.

\bibitem[Hess and de~Moraes~Russo(2007)]{Hess:07}
Pablo~N Hess and Claudia~A de~Moraes~Russo.
\newblock An empirical test of the midpoint rooting method.
\newblock \emph{Biol. J. Linnean Soc.}, 92:\penalty0 669--674, 2007.
\newblock \doi{10.1111/j.1095-8312.2007.00864.x}.

\bibitem[Hittinger and Carroll(2007)]{Hittinger:07}
C~T Hittinger and S~B Carroll.
\newblock Gene duplication and the adaptive evolution of a classic genetic
  switch.
\newblock \emph{Nature}, 449:\penalty0 677--681, 2007.
\newblock \doi{10.1038/nature06151}.

\bibitem[Holland et~al.(2003)Holland, Penny, and Hendy]{Holland:03}
B~R Holland, D~Penny, and M~D Hendy.
\newblock Outgroup misplacement and phylogenetic inaccuracy under a molecular
  clock --- a simulation study.
\newblock \emph{Syst. Biol}, 52:\penalty0 229--238, 2003.
\newblock \doi{10.1080/10635150390192771}.

\bibitem[Huelsenbeck et~al.(2002)Huelsenbeck, Larget, Miller, and
  Ronquist]{Huelsenbeck:02}
J~P Huelsenbeck, B~Larget, R~E Miller, and F~Ronquist.
\newblock Potential applications and pitfalls of {Bayesian} inference of
  phylogeny.
\newblock \emph{Syst Biol}, 51:\penalty0 673--688, 2002.
\newblock \doi{10.1080/10635150290102366}.

\bibitem[Jukes and Cantor(1969)]{Jukes:69}
T~H Jukes and C~R Cantor.
\newblock Evolution of protein molecules.
\newblock In H~N Munro, editor, \emph{Mammalian Protein Metabolism}, pages
  21--132. Academic Press, New York, 1969.

\bibitem[Katz et~al.(2012)Katz, Grant, Parfrey, and Burleigh]{Katz:12}
L~A Katz, J~R Grant, L~W Parfrey, and J~G Burleigh.
\newblock Turning the crown upside down: gene tree parsimony roots the
  eukaryotic tree of life.
\newblock \emph{Syst. Biol.}, 61:\penalty0 653--660, 2012.
\newblock \doi{10.1093/sysbio/sys026}.

\bibitem[Kawahara and Imanishi(2007)]{Kawahara:07}
Y~Kawahara and T.~Imanishi.
\newblock A genome-wide survey of changes in protein evolutionary rates across
  four closely related species of \emph{Saccharomyces} sensu stricto group.
\newblock \emph{BMC Evol Biol.}, 7:\penalty0 9, 2007.
\newblock \doi{10.1186/1471-2148-7-9}.

\bibitem[Keller-Schmidt and Klemm(2012)]{Keller:2012}
Stephanie Keller-Schmidt and Konstantin Klemm.
\newblock A model of macroevolution as a branching process based on
  innovations.
\newblock \emph{Adv. Complex Syst.}, 15:\penalty0 1250043, 2012.
\newblock \doi{10.1142/S0219525912500439}.

\bibitem[Kimura(1980)]{Kimura:80}
M~Kimura.
\newblock A simple method for estimating evolutionary rates of base
  substitutions through comparative studies of nucleotide sequences.
\newblock \emph{J. Mol. Evol.}, 16:\penalty0 111--120, 1980.
\newblock \doi{10.1007/BF01731581}.

\bibitem[Kinene et~al.(2016)Kinene, Wainaina, Maina, and Boykin]{Kinene:16}
T~Kinene, J~Wainaina, S~Maina, and LM~Boykin.
\newblock Rooting trees, methods for.
\newblock In Richard~M. Kliman, editor, \emph{Encyclopedia of Evolutionary
  Biology}, volume~3, page 489. Elsevier, Amsterdam, NL, 2016.
\newblock \doi{10.1016/B978-0-12-800049-6.00215-8}.

\bibitem[Koonin(2005)]{Koonin:05}
Eugene Koonin.
\newblock Orthologs, paralogs, and evolutionary genomics.
\newblock \emph{Ann. Rev. Genetics}, 39:\penalty0 309--338, 2005.
\newblock \doi{10.1146/annurev.genet.39.073003.114725}.

\bibitem[Kumar(2005)]{Kumar:05}
Sudhir Kumar.
\newblock Molecular clocks: four decades of evolution.
\newblock \emph{Nat Rev Genet}, 6:\penalty0 654--662, 2005.
\newblock \doi{10.1038/nrg1659. PMID 16136655}.

\bibitem[Lechner et~al.(2011)Lechner, Findei{\ss}, Steiner, Marz, Stadler, and
  Prohaska]{Lechner:11a}
Marcus Lechner, Sven Findei{\ss}, Lydia Steiner, Manja Marz, Peter~F. Stadler,
  and Sonja~J. Prohaska.
\newblock \texttt{Proteinortho:} detection of (co-)orthologs in large-scale
  analysis.
\newblock \emph{BMC Bioinformatics}, 12:\penalty0 124, 2011.
\newblock \doi{10.1186/1471-2105-12-124}.

\bibitem[Mai et~al.()Mai, Sayyari, and Mirarab]{Mai:17}
Uyen Mai, Erfan Sayyari, and Siavash Mirarab.
\newblock Minimum variance rooting of phylogenetic trees and implications for
  species tree reconstruction.
\newblock \emph{PLoS ONE}, 12:\penalty0 e0182238.
\newblock \doi{10.1371/journal.pone.0182238}.

\bibitem[Nieselt-Struwe(1997)]{Nieselt:97}
K.~Nieselt-Struwe.
\newblock Graphs in sequence spaces: a review of statistical geometry.
\newblock \emph{Biophys Chem.}, 66:\penalty0 111--131, 1997.
\newblock \doi{10.1016/S0301-4622(97)00064-1}.

\bibitem[Nieselt-Struwe and von Haeseler(2001)]{Nieselt:01}
Kay Nieselt-Struwe and Arndt von Haeseler.
\newblock Quartet-mapping, a generalization of the likelihood-mapping
  procedure.
\newblock \emph{Mol Biol Evol}, 18:\penalty0 1204--1219, 2001.
\newblock \doi{10.1093/oxfordjournals.molbev.a003907}.

\bibitem[Overbeek et~al.(1999)Overbeek, Fonstein, D'Souza, Pusch, and
  Maltsev]{Overbeek:99}
R~Overbeek, M~Fonstein, M~D'Souza, G~D Pusch, and N~Maltsev.
\newblock The use of gene clusters to infer functional coupling.
\newblock \emph{Proc Natl Acad Sci USA}, 96:\penalty0 2896--2901, 1999.
\newblock \doi{10.1073/pnas.96.6.2896}.

\bibitem[Penny(1976)]{Penny:76}
David Penny.
\newblock Criteria for optimising phylogenetic trees and the problem of
  determining the root of a tree.
\newblock \emph{J. Mol. Evol.}, 8:\penalty0 95--116, 1976.
\newblock \doi{10.1007/BF01739097}.

\bibitem[Retzlaff and Stadler(2018)]{Retzlaff:18b}
Nancy Retzlaff and Peter~F. Stadler.
\newblock Phylogenetics beyond biology.
\newblock \emph{Th. Biosci.}, 137:\penalty0 133--143, 2018.
\newblock \doi{10.1007/s12064-018-0264-7}.

\bibitem[Rusin et~al.(2014)Rusin, Lyubetskaya, Gorbunov, and
  Lyubetsky]{Rusin:14}
L.~Y. Rusin, E.~Lyubetskaya, K.~Y. Gorbunov, and V.~Lyubetsky.
\newblock Reconciliation of gene and species trees.
\newblock \emph{BioMed Res Int.}, 2014:\penalty0 642089, 2014.
\newblock \doi{10.1155/2014/642089}.

\bibitem[Saitou and Nei(1987)]{Saitou:87}
N.~Saitou and M.~Nei.
\newblock The neighbor-joining method: a new method for reconstructing
  phylogenetic trees.
\newblock \emph{Mol Biol Evol}, 4:\penalty0 406--425, 1987.
\newblock \doi{10.1093/oxfordjournals.molbev.a040454}.

\bibitem[Sattah and Tversky(1977)]{Sattah:77}
S.~Sattah and A.~Tversky.
\newblock Additive similarity trees.
\newblock \emph{Psychometrika}, 42:\penalty0 319--345, 1977.
\newblock \doi{10.1007/BF02293654}.

\bibitem[Semple and Steel(2003)]{Semple:03a}
Charles Semple and Mike Steel.
\newblock \emph{Phylogenetics}.
\newblock Oxford University Press, Oxford UK, 2003.

\bibitem[Shavit et~al.(2007)Shavit, Penny, Hendy, and Holland]{Shavit:07}
L~Shavit, D~Penny, M~D Hendy, and B~R Holland.
\newblock The problem of rooting rapid radiations.
\newblock \emph{Mol Biol Evol}, 24:\penalty0 2400--2411, 2007.
\newblock \doi{10.1093/molbev/msm178}.

\bibitem[Sim{\~o}es-Pereira(1969)]{SimoesPereira:69}
J.~M.~S. Sim{\~o}es-Pereira.
\newblock A note on the tree realizability of a distance matrix.
\newblock \emph{J. Combin. Theory}, 6:\penalty0 303--310, 1969.
\newblock \doi{10.1016/S0021-9800(69)80092-X}.

\bibitem[Soria et~al.(2014)Soria, McGary, and Rokas]{Soria:14}
P~S Soria, K~L McGary, and A.~Rokas.
\newblock Functional divergence for every paralog.
\newblock \emph{Mol Biol Evol.}, 31:\penalty0 984--992, 2014.
\newblock \doi{10.1093/molbev/msu050}.

\bibitem[Spielman and Wilke(2015)]{Spielman:15}
S~J Spielman and C~O Wilke.
\newblock {Pyvolve}: A flexible python module for simulating sequences along
  phylogenies.
\newblock \emph{PLoS One}, 10:\penalty0 e0139047, 2015.
\newblock \doi{10.1371/journal.pone.0139047}.

\bibitem[Steel(1992)]{Steel:92}
Mike Steel.
\newblock The complexity of reconstructing trees from qualitative characters
  and subtress.
\newblock \emph{J. Classification}, 9:\penalty0 91--116, 1992.

\bibitem[Swofford et~al.(1996)Swofford, Olsen, Waddell, and
  Hillis]{Swofford:96}
D~L Swofford, G~J Olsen, P~J Waddell, and D~M Hillis.
\newblock Phylogenetic inference.
\newblock In D~M Hillis, C~Moritz, and B~K Mable, editors, \emph{Molecular
  systematics}, pages 407--514. Sinauer Associates, Sunderland, MA, 1996.

\bibitem[Tamura(1992)]{Tamura:92}
K~Tamura.
\newblock Estimation of the number of nucleotide substitutions when there are
  strong transition-transversion and {G}+{C} content biases.
\newblock \emph{Mol. Biol. Evol.}, 9:\penalty0 678--687, 1992.
\newblock \doi{10.1093/oxfordjournals.molbev.a040752}.

\bibitem[Tatusov et~al.(1997)Tatusov, Koonin, and Lipman]{Tatusov:97}
R.~L. Tatusov, E.~V. Koonin, and D.~J. Lipman.
\newblock A genomic perspective on protein families.
\newblock \emph{Science}, 278:\penalty0 631--637, 1997.
\newblock \doi{10.1126/science.278.5338.631}.

\bibitem[Train et~al.()Train, Glover, Gonnet, Altenhoff, and
  Dessimoz]{Train:17}
Cl{\'e}ment-Marie Train, Natasha~M Glover, Gaston~H Gonnet, Adrian~M Altenhoff,
  and Christophe Dessimoz.
\newblock Orthologous matrix ({OMA}) algorithm 2.0: more robust to asymmetric
  evolutionary rates and more scalable hierarchical orthologous group
  inference.
\newblock \emph{Bioinformatics}, 33:\penalty0 i75--i82.
\newblock \doi{10.1093/bioinformatics/btx229}.

\bibitem[Wagner et~al.(2005)Wagner, Takahashi, Lynch, Prohaska, Fried, Stadler,
  and Amemiya]{Wagner:05a}
G{\"u}nter~P. Wagner, Kazuhiko Takahashi, Vincent Lynch, Sonja~J. Prohaska,
  Claudia Fried, Peter~F. Stadler, and Chris~T. Amemiya.
\newblock Molecular evolution of duplicated ray finned fisch hoxa clusters:
  Increased synonymous substitution rate and asymmetrical co-divergence of
  coding and non-coding sequences.
\newblock \emph{J.\ Mol.\ Evol.}, pages 665--676, 2005.

\bibitem[Wall et~al.(2003)Wall, Fraser, and Hirsh]{Wall:03}
D~P Wall, H~B Fraser, and A~E Hirsh.
\newblock Detecting putative orthologs.
\newblock \emph{Bioinformatics}, 19:\penalty0 1710--1711, 2003.
\newblock \doi{10.1093/bioinformatics/btg213}.

\bibitem[Whelan and Goldman(2001)]{Whelan:01}
S~Whelan and N~Goldman.
\newblock A general empirical model of protein evolution derived from multiple
  protein families using a maximum-likelihood approach.
\newblock \emph{Mol Biol Evol}, 18:\penalty0 691--699, 2001.
\newblock \doi{10.1093/oxfordjournals.molbev.a003851}.

\bibitem[Williams et~al.(2015)Williams, Heaps, Cherlin, Nye, Boys, and
  Embley]{Williams:15}
Tom~A. Williams, Sarah~E. Heaps, Svetlana Cherlin, Tom M.~W. Nye, Richard~J.
  Boys, and T.~Martin Embley.
\newblock New substitution models for rooting phylogenetic trees.
\newblock \emph{Philos Trans R Soc Lond B Biol Sci}, 370:\penalty0 20140336,
  2015.
\newblock \doi{10.1098/rstb.2014.0336}.

\bibitem[Yu et~al.(2011)Yu, Zavaljevski, Desai, and Reifman]{Yu:11}
Chenggang Yu, Nela Zavaljevski, Valmik Desai, and Jaques Reifman.
\newblock {QuartetS}: a fast and accurate algorithm for large-scale orthology
  detection.
\newblock \emph{Nucleic Acids Res.}, 39:\penalty0 e88, 2011.
\newblock \doi{10.1093/nar/gkr308}.

\bibitem[Zuckerkandl and Pauling(1962)]{Zuckerkandl:62}
E.~Zuckerkandl and L.~B. Pauling.
\newblock Molecular disease, evolution, and genic heterogeneity.
\newblock In M.~Kasha and B~Pullman, editors, \emph{Horizons in Biochemistry},
  pages 189--225. Academic Press, New York, 1962.

\end{thebibliography}

\end{document}